\RequirePackage{fix-cm}
\documentclass[smallextended,draft]{svjour3}       % onecolumn (second format)
\smartqed  % flush right qed marks, e.g. at end of proof
\usepackage{graphicx}
\usepackage{mathptmx}      % use Times fonts if available on your TeX system
\usepackage{latexsym,amsfonts,amssymb,amsmath,tensor}
\usepackage{marvosym}
\newcommand{\envelope}{(\raisebox{-.5pt}{\scalebox{1.45}{\Letter}}\kern-1.7pt)}

\newcommand{\rmd}{\textrm d}
\newcommand{\rmi}{\textrm i}
\journalname{General Relativity and Gravitation}

\begin{document}
\title{PP-waves with Torsion - a Metric-affine Model for the Massless Neutrino}

\author{Vedad Pasic \and Elvis Barakovic}

\institute{V. Pasic \envelope \at Department of Mathematics, University of Tuzla, Univerzitetska 4, 75000 Tuzla, Bosnia and Herzegovina \\
Tel.: +387-61-195464 \columncase{,}{\\}
Fax: +387-35-320861\columncase{,}{\\}
\email{vedad.pasic@untz.ba} \and
E. Barakovic  \at Department of Mathematics, University of Tuzla, Univerzitetska 4, 75000 Tuzla, Bosnia and Herzegovina \\ \email{elvis.barakovic@untz.ba}
}

\date{Received: date / Accepted: date}
%\titlerunning{Metric-affine Model for the Massless Neutrino}
%\authorrunning{E. Barakovic and V. Pasic}

\maketitle

\begin{abstract}
In this paper we deal with quadratic metric--affine gravity, which we briefly introduce,
explain and give historical and physical reasons for using this particular theory of gravity.
We then introduce a generalisation of well known spacetimes, namely pp-waves.
A classical pp-wave is a 4-dimensional Lorentzian spacetime \,
which admits a nonvanishing parallel spinor field;
here the connection is assumed to be Levi-Civita.
This definition was generalised in our previous work to metric compatible
spacetimes with torsion and used to construct new explicit vacuum solutions of
quadratic metric--affine gravity, namely \emph{generalised} pp-waves of parallel Ricci curvature.
The physical interpretation of these solutions we propose in this
article is that they represent a \emph{conformally invariant
metric--affine model for a massless elementary particle}. We give a comparison with the
classical model describing the interaction of gravitational and
massless neutrino fields, namely \emph{Einstein--Weyl theory} and construct pp-wave type
solutions of this theory. We point out that generalised pp-waves of parallel Ricci curvature are very
similar to pp-wave type solutions of the Einstein--Weyl model and therefore propose that
our generalised pp-waves of parallel Ricci curvature represent a metric--affine model for the massless
neutrino.
\end{abstract}

\section{Introduction}
\label{introduction}

The smallest departure from a Riemannian spacetime of Einstein's general
relativity would consist of admitting \emph{torsion} (\ref{definitiontorsion}),
arriving thereby at a Riemann--Cartan spacetime, and, furthermore, possible nonmetricity
(\ref{nonmetricity}), resulting in a `\emph{metric--affine}' spacetime. Metric--affine gravity is a natural generalisation of Einstein's general relativity, which is based on a spacetime with a Riemannian metric $g$ of Lorentzian signature.

We consider spacetime to be a connected real
4-manifold $M$ equipped with a Lorentzian metric $g$
and an affine connection $\Gamma$.
The 10 independent
components of the (symmetric) metric tensor $g_{\mu\nu}$
and the 64 connection coefficients ${\Gamma^\lambda}_{\mu\nu}$
are the unknowns of our theory. Note that the characterisation of the spacetime manifold
by an \emph{independent} linear connection $\Gamma$ initially
distinguishes metric--affine gravity from general relativity. The
connection incorporates the inertial properties of spacetime and it
can be viewed, according to Hermann Weyl \cite{weylquadraticaction},
as the guidance field of spacetime. The metric describes the structure of spacetime with
respect to its spacio-temporal distance relations.

According to Hehl et al. in \cite{hehlreview}, in Einstein's general relativity
the linear connection of its Riemannian spacetime is metric--compatible
and symmetric.
The symmetry of the Levi--Civita connection translates into the closure of
infinitesimally small parallelograms.
Already the transition from the flat gravity-free Minkowski
spacetime to the Riemannian spacetime in Einstein's theory can locally be
understood as a deformation process. The lifting of the constraints of metric--compatibility and symmetry
yields nonmetricity and torsion, respectively.
The continuum under consideration, here classical spacetime, is thereby assumed to have
a non-trivial microstructure, similar to that of a liquid crystal or a dislocated metal or
feromagnetic material etc.
It is gratifying to have the geometrical concepts
of nonmetricity and torsion already arising in the (three-dimensional) continuum theory of lattice
defects, see \cite{KronerLattices,KronerKristal}.

We define our action as
\begin{equation}
\label{action}
S:=\int q(\!R)
\end{equation}
where $q$ is an $\mathrm{O}(1,3)$-invariant
quadratic form on curvature $R\,$.
Independent variation of the metric
$g$ and the connection $\Gamma$ produces Euler--Lagrange
equations which we will write symbolically as
\begin{eqnarray}
\label{eulerlagrangemetric}
\partial S/\partial g&=&0,
\\
\label{eulerlagrangeconnection}
\partial S/\partial\Gamma&=&0.
\end{eqnarray}
Our objective is the study of the
combined system of field equations
(\ref{eulerlagrangemetric}), (\ref{eulerlagrangeconnection}).
This is a system of $10+64$
real nonlinear partial differential equations
with $10+64$ real unknowns.

Our motivation comes from
Yang--Mills theory. The Yang--Mills action for the affine connection
is a special case of (\ref{action}) with
\begin{equation}
\label{YMq} q(\!R)=q_{\mathrm{YM}}(\!R):=
R^\kappa{}_{\lambda\mu\nu}\, R^\lambda{}_\kappa{}^{\mu\nu}\,.
\end{equation}
With this choice of $q(R)$, equation (\ref{eulerlagrangeconnection}) is
the Yang--Mills equation for the affine connection,
which was analysed by Yang \cite{yang}. The quadratic form $q$ appearing in (\ref{action}) is a generalisation of (\ref{YMq}). The general formula for $q$ contains 16 different $R^2$-terms
with 16 coupling constants. This formula is given in Appendix B
of \cite{annalen}. An equivalent formula can be found in
\cite{Esser,hehlandmaciasexactsolutions2}.

Yang was looking for Riemannian
solutions, so he specialised equation
(\ref{eulerlagrangeconnection}) to the Levi--Civita connection
\begin{equation}
\label{eulerlagrangeconnectionLC} \left.\partial
S/\partial\Gamma\,\right|_{\mathrm{L-C}}=0.
\end{equation} Here `specialisation' means that one sets
${\Gamma^\lambda}_{\mu\nu}=\genfrac{\{}{\}}{0pt}{}{\lambda}{\mu\nu}$
\emph{after} the variation in $\Gamma$ is carried out. It is known \cite{pseudo} that in this case for a generic 11-parameter action equation (\ref{eulerlagrangeconnectionLC}) reduces to
\begin{equation}
\label{compactYangMils}
\nabla_\lambda Ric_{\kappa\mu}
-\nabla_\kappa Ric_{\lambda\mu}=0.
\end{equation}
However,
according to \cite{annalen} for a generic 16-parameter
action equation (\ref{eulerlagrangeconnectionLC}) reduces to
\begin{equation}
\label{eulerlagrangeconnectionLC16} \nabla Ric=0.
\end{equation}
The field equations (\ref{compactYangMils}) and
(\ref{eulerlagrangeconnectionLC16}) are very much different, with
(\ref{eulerlagrangeconnectionLC16}) being by far more restrictive.
In particular, Nordstr\"om--Thompson spacetimes (Riemannian
spacetimes with ${}^*\!R^*=+R$) satisfy
(\ref{compactYangMils}) but do not necessarily satisfy
(\ref{eulerlagrangeconnectionLC16}). Note that the LHS of (\ref{compactYangMils}) is the Cotton tensor for metric compatible spacetimes with constant scalar curvature, see e.g. \cite{cotton,garciaCotton}. Compare this to equation (8) from \cite{annalen} and equation (27) from \cite{pasicvassiliev}.

Let us here mention the contributions of C.N. Yang \cite{yang} and E.W. Mielke \cite{mielkepseudoparticle} who showed, respectively, that Einstein spaces satisfy equations (\ref{eulerlagrangeconnection}) and (\ref{eulerlagrangemetric}).There is a substantial bibliography devoted to the study of the system (\ref{eulerlagrangemetric}), (\ref{eulerlagrangeconnection}) in the special case (\ref{YMq}) and one can get an idea of the historical development of the Yang--Mielke theory of gravity from \cite{buchdahl,fairchild1976,fairchild1976erratum,olesen,pavelleApr1975,stephenson,thompsonFeb1975,thompsonAug1975,wilczek}.

The motivation for choosing a model of gravity which is purely quadratic
in curvature is explained in Section 1 of \cite{annalen}.
The study of equations (\ref{eulerlagrangemetric}),
(\ref{eulerlagrangeconnection}) for specific purely quadratic curvature
Lagrangians has a long history. Quadratic  curvature Lagrangians were first discussed
by Weyl \cite{weylquadraticaction}, Pauli \cite{pauli}, Eddington \cite{eddington} and Lanczos
\cite{lanczos1,lanczos2,lanczos3} in an attempt to include the electromagnetic
field in Riemannian geometry.

The idea of using a purely quadratic action in General Relativity
goes back to Hermann Weyl, as given at the end of his paper
\cite{weylquadraticaction}, where he argued that the most
natural gravitational action should be quadratic in curvature and
involve all possible invariant quadratic combinations of curvature,
like the square of Ricci curvature, the square of scalar curvature,
etc. Unfortunately, Weyl himself never afterwards pursued this
analysis. Stephenson \cite{stephenson} looked at three different
quadratic invariants: scalar curvature squared, Ricci curvature
squared and the Yang--Mills Lagrangian (\ref{YMq}) and varied with respect
to the metric and the affine connection.
He concluded that every equation arising from the
above mentioned quadratic Lagrangians has the Schwarzschild solution and that
the equations give the same results for the three `crucial tests'
of general relativity, i.e. the bending of light, the advance of the
perihelion of Mercury and the red-shift.

Higgs \cite{higgs} continued in a similar fashion to show that in scalar squared
and Ricci squared cases, one set of equations may be transformed into
field equations of the Einstein type with an arbitrary `cosmological
constant' in terms of the `new gauge--invariant metric'.

One can get more information and form an idea on the historical development
of the quadratic metric--affine theory of gravity from
\cite{buchdahl,fairchild1976,fairchild1976erratum,higgs,mielkepseudoparticle,olesen,pavelleApr1975,stephenson,thompsonFeb1975,thompsonAug1975,wilczek,yang}.

It should be noted that the action (\ref{action}) contains only
purely quadratic curvature terms, so it excludes the
Einstein--Hilbert term (linear in curvature) and any terms quadratic
in torsion (\ref{definitiontorsion}) and nonmetricity (\ref{nonmetricity}).
By choosing a purely quadratic curvature
Lagrangian we are hoping to describe phenomena whose characteristic
wavelength is sufficiently small and curvature sufficiently large.

We should also point out that the
action (\ref{action}) is conformally invariant, i.e. it does not
change if we perform a Weyl rescaling of the metric $g\to e^{2f}g$,
$f:M\to\mathbb{R}$, without changing the connection $\Gamma$.

In our previous work \cite{pasicvassiliev}, we presented new non-Riemannian
solutions of the field equations
(\ref{eulerlagrangemetric}), (\ref{eulerlagrangeconnection}).
These new solutions were to be constructed explicitly
and the construction turned out to be very similar
to the classical construction of a pp-wave, only with torsion. This paper aims to
give additional information on these spacetimes, provide their physical interpretation,
additional calculations and future possible applications.

The paper has the following structure.
In section \ref{pptorsionChapter} we recall basic facts about pp-waves with torsion,
in subsection \ref{Classical pp-waves} we
provide information about classical pp-waves, in subsection \ref{PP-waves with torsion}
we recall the way pp-waves were generalised in \cite{pasicvassiliev} and list the properties of these
spacetimes with torsion and in subsection \ref{SpinorFormalismPPWaves} we present the pp-waves spinor formalism.
In section \ref{interpretation} we present our attempt at giving a physical interpretation to the solutions of the field equations from \cite{pasicvassiliev}. In subsection \ref{EinsteinWeylFieldEquations} we provide
a reminder on the classical model describing the interaction of gravitational and massless neutrino fields (Einstein--Weyl theory), while in subsection \ref{EWComparisonExisting} we present a brief review of known solutions of this theory.
In subsection \ref{ppwaveEW} we present our pp-wave type solutions of this theory.
In subsection \ref{comparisonSection} we compare the Einstein--Weyl solutions to our conformally invariant solutions. Finally, appendix \ref{AppendixSpinorFormalism} provides the spinor formalism used throughout our work, appendix \ref{AppendixWeyl} gives detailed calculations involved in comparing our solutions to
Einstein--Weyl theory and appendix \ref{AppendixCorrection} provides a correction of a mistake
found in our previous work \cite{pasicvassiliev} and gives the motivation for future work.

\section{Notation}
\label{Notation}

Our notation follows \cite{King and Vassiliev,skopje,pasicvassiliev,pseudo,annalen}. In particular,
we denote local coordinates by $x^\mu$, $\mu=0,1,2,3$,
and write $\partial_\mu:=\partial/\partial x^\mu$.
We define torsion as
\begin{equation}\label{definitiontorsion}
{T^\lambda}_{\mu\nu}:=
{\Gamma^\lambda}_{\mu\nu}-{\Gamma^\lambda}_{\nu\mu}\,.
\end{equation}
The irreducible pieces of torsion are, following \cite{pseudo},
\begin{equation}\label{torsion-pieces1}
T^{(1)}=T-T^{(2)}-T^{(3)},\quad \tensor{{T^{(2)}}}{_{\lambda\mu\nu}}= g_{\lambda\mu}v_\nu-g_{\lambda\nu}v_\mu, \quad
T^{(3)}=\ast w,
\end{equation}
where
\begin{equation}\label{torsion-pieces2}
v_\nu=\frac{1}{3}\tensor{T}{^{\lambda}_{\lambda\nu}},\ \ w_\nu=\frac{1}{6}\sqrt{|\det g|}\tensor{T}{^{\kappa\lambda\mu}}\tensor{\varepsilon}{_{\kappa\lambda\mu\nu}}.
\end{equation}
The pieces $T^{(1)},T^{(2)}$ i $T^{(3)}$ are called \emph{tensor torsion}, \emph{trace torsion}, and \emph{axial torsion}
respectively. We say that our connection $\Gamma$ is metric compatible
if $\nabla g\equiv0$. The interval is $\rmd s^2:=g_{\mu\nu}\,d x^\mu\,d x^\nu$. Given a scalar function $f:M\to\mathbb{R}$ we write for brevity
$\displaystyle
\int f:=\int_Mf\,\sqrt{|\det g|}\,\rmd x^0\rmd x^1\rmd x^2\rmd x^3\,,
\det g:=\det(g_{\mu\nu})\,.
$ We define \emph{nonmetricity} by
\begin{equation}\label{nonmetricity}
Q_{\mu \alpha\beta}:=\nabla_\mu g_{\alpha\beta}.
\end{equation}
We use the term `parallel'
to describe the situation when the covariant derivative of some spinor
or tensor field is identically zero. We do not assume that our spacetime admits a (global) spin structure,
cf. Section 11.6 of \cite{Nakahara}.
In fact, our only topological assumption is connectedness.
This does not prevent us from defining and parallel transporting
spinors or tensors locally.

\section{PP-waves With Torsion}\label{pptorsionChapter}

In this section, where we mostly follow the exposition from \cite{pasicvassiliev}, we provide background information about \emph{pp-waves}, starting with the notion
of a classical pp-wave, then introducing a generalisation with the addition of torsion and lastly presenting the particular spinor formalism of pp-waves.

%%%%%%%%%%%%%%%%%%%%%%%%%%%%%%%%%%
%      CLASSIC PP WAVES          %
%%%%%%%%%%%%%%%%%%%%%%%%%%%%%%%%%%

\subsection{Classical pp-waves}\label{Classical pp-waves}

PP-waves are well known spacetimes in general relativity, first
discovered by Brink--mann \cite{brinkmann} in 1923, and subsequently
rediscovered by several authors, for example Peres
\cite{peres} in 1959.
There are differing views on what the `pp' stands for. According to
Griffiths \cite{griffiths book} and Kramer et al. \cite{exact solutions Kramer}
`pp' is an abbreviation for
`plane-fronted gravitational waves with parallel rays'.
See e.g. \cite{Adamowitz,Babourova,blagojevichehl,brinkmann,garcia,griffiths book,exact solutions Kramer,obukhov1,obukhov2,balkanica,pasicvassiliev,peres,peresweb,poland,annalen} for more information on pp-waves and pp-wave type solutions of metric--affine gravity.

\begin{definition}\label{definition 1 of a pp-space}
A \emph{pp-wave} is a Riemannian spacetime which admits a
nonvanishing parallel spinor field.
\end{definition}
It was only relatively recently discovered in \cite{annalen} that pp-waves of parallel Ricci curvature are solutions of (\ref{eulerlagrangemetric}),
(\ref{eulerlagrangeconnection}). This section closely follows Section 3 from \cite{pasicvassiliev} in exposition and we only present the the most important facts about these well known spacetimes.
The nonvanishing parallel spinor field appearing in the definition
of pp-waves will be denoted throughout this paper by
\[
\chi=\chi^a
\] and we assume this spinor field to be \emph{fixed}.
Put
\begin{equation}
\label{formula for l} l^\alpha:=\sigma^\alpha{}_{a\dot
b}\,\chi^a\bar\chi^{\dot b}
\end{equation}
where the $\sigma^\alpha$ are Pauli matrices\footnote{See appendix
\ref{AppendixSpinorFormalism} for our general spinor formalism and section
\ref{SpinorFormalismPPWaves} for spinor formalism for pp-waves}. Then
$l$ is a nonvanishing parallel real null vector field. Now we define the real scalar function
\begin{equation}
\label{phase} \varphi:M\to\mathbb{R},\qquad \varphi(x):=\int l\cdot
d x\,.
\end{equation}
This function is called the \emph{phase}. It is defined uniquely up
to the addition of a constant and possible multi-valuedness resulting
from a nontrivial topology of the manifold.
Put
\begin{equation}
\label{formula for F} F_{\alpha\beta}:=\sigma_{\alpha\beta
ab}\,\chi^a\chi^b
\end{equation}
where the $\sigma_{\alpha\beta}$ are `second order Pauli matrices'
(\ref{second order Pauli matrices for pp-metric}), (\ref{second order Pauli matrices}). Then $F$ is a nonvanishing
parallel complex 2-form with the additional properties $*F=\pm\rmi
F$ and $\det F=0$. It can be written as
\begin{equation}
\label{decomposition of F} F=l\wedge m
\end{equation}
where $m$ is a complex vector field satisfying $m_\alpha
m^\alpha=l_\alpha m^\alpha=0$, $m_\alpha\bar{m}^\alpha=-2$.
It is known, see \cite{Alekseevsky,Bryant}, that Definition
\ref{definition 1 of a pp-space} is equivalent to the following

\begin{definition}
\label{definition 2 of a pp-space} A \,\emph{pp-wave} is a
Riemannian spacetime whose metric can be written locally in the form
\begin{equation}
\label{metric of a pp-wave} d s^2= \,2\,d x^0\,d x^3-(d x^1)^2-(d
x^2)^2 +f(x^1,x^2,x^3)\,(d x^3)^2
\end{equation}
in some local coordinates $(x^0,x^1,x^2,x^3)$.
\end{definition}
The remarkable property of the metric
(\ref{metric of a pp-wave}) is that the corresponding curvature
tensor $R$ is linear in $f$:
\begin{equation}
\label{curvature of a pp-space special} R_{\alpha\beta\gamma\delta}=
-\frac12(l\wedge\partial)_{\alpha\beta}\,(l\wedge\partial)_{\gamma\delta}f
\end{equation}
where
$(l\wedge\partial)_{\alpha\beta}:=l_\alpha\partial_\beta-\partial_\alpha
l_\beta$.
The advantage of Definition \ref{definition 2 of a pp-space} is that
it gives an explicit formula for the metric of a pp-wave. Its
disadvantage is that it relies on a particular choice of local
coordinates in each coordinate patch.
The choice of local coordinates in which the pp-metric assumes the form
(\ref{metric of a pp-wave}) is not unique. We will restrict our
choice to those coordinates in which
\begin{equation}
\label{explicit l and a} \chi^a=(1,0), \qquad l^\mu=(1,0,0,0),
\qquad m^\mu=(0,1,\mp\rmi,0).
\end{equation}
With such a choice formula (\ref{phase}) reads
$\varphi(x)=x^3+\mathrm{const}$.
Formula (\ref{curvature of a pp-space special}) can be rewritten
in invariant form
\begin{equation}
\label{curvature of a pp-space}
R=-\frac12(l\wedge\nabla)\otimes(l\wedge\nabla)f
\end{equation}
where $l\wedge\nabla:=l\otimes\nabla-\nabla\otimes l$.
The curvature of a pp-wave has the following irreducible pieces:
(symmetric) trace--free Ricci and Weyl.
Ricci curvature is proportional to $l\otimes l$ whereas Weyl
curvature is a linear combination of $\mathrm{Re}\left((l\wedge
m)\otimes(l\wedge m)\right)$ and $\mathrm{Im}\left((l\wedge
m)\otimes(l\wedge m)\right)$. In our special local coordinates (\ref{metric of a pp-wave}),
(\ref{explicit l and a}), we can express these as
\begin{eqnarray}
\label{explicit W and Ric1}
Ric &=& \frac12 (f_{11} + f_{22}) \, l \otimes l,
\\
\label{explicit W and Ric2}
\mathcal{W} &=& \sum_{j,k=1}^2 w_{jk} (l\wedge m_j)\otimes(l\wedge m_k),
\end{eqnarray}
where $m_1=\mathrm{Re } (m), m_2=\mathrm{Im }(m)$, $f_{\alpha\beta} := \partial_{\alpha}\partial_{\beta} f$ and $w_{jk}$ are real scalars given by
$\displaystyle w_{11} = \frac14 (-f_{11} +f_{22}), \ w_{12} = \pm \frac12 f_{12}, \ w_{22} = -w_{11},  \ w_{21} = w_{12}$.
\begin{remark}
Note that the Cotton tensor of classical pp-waves with parallel Ricci curvature vanishes, as classical pp-waves are metric compatible spacetimes with zero scalar curvature. In the theory of conformal spaces the main geometrical objects to be analysed are the Weyl and the Cotton tensors, see \cite{garciaCotton}.
It is well known that for conformally flat spaces the Weyl tensor has to vanish and consequently the Cotton tensor has to vanish too. Note that the Cotton tensor is only conformally invariant in three dimensions.
\end{remark}

%%%%%%%%%%%%%%%%%%%%%%%%%%%%%%%%%%
%      GENERALISED PP-WAVES      %
%%%%%%%%%%%%%%%%%%%%%%%%%%%%%%%%%%

\subsection{Generalised pp--waves}
\label{PP-waves with torsion}

One natural way of generalising the concept of a classical
pp-wave is simply to extend Definition \ref{definition 1 of a
pp-space} to general metric compatible spacetimes, i.e. spacetimes
whose connection is not necessarily Levi-Civita. However, this gives
a class of spacetimes which is too wide and difficult to work with.
We choose to extend the classical definition in a more special way
better suited to the study of the system of field equations
(\ref{eulerlagrangemetric}), (\ref{eulerlagrangeconnection}).

Consider the polarized Maxwell equation\footnote{See in particular \cite{King and Vassiliev} as well as \cite{balkanica,pasicvassiliev,pseudo,garda,poland,annalen}}
\begin{equation}
\label{polarized Maxwell equation}
*d A=\pm\rmi d A
\end{equation}
in a classical pp-space. Here
$A$ is the unknown complex vector field. The motivation for calling equation (\ref{polarized Maxwell equation}) the polarized Maxwell equation comes from the fact that any solution of (\ref{polarized Maxwell equation}) is a solution of the Maxwell equation $\delta \mathrm{d}u=0$, see \cite{King and Vassiliev}.
We seek plane wave
solutions of (\ref{polarized Maxwell equation}). These can be written down
explicitly:
\begin{equation}
\label{plane wave} A=h(\varphi)\,m\,+\,k(\varphi)\,l\,.
\end{equation}
Here $l$ and $m$ are the vector fields defined in section
\ref{Classical pp-waves}, $\ h,k:\mathbb{R}\to\mathbb{C}$ are
arbitrary functions, and $\varphi$ is the phase (\ref{phase}).

\begin{definition}
\label{definition of a generalised pp-space} A \emph{generalised
pp-wave} is a metric compatible spacetime with pp-metric and torsion
\begin{equation}
\label{define torsion} T:=\frac12\mathrm{Re}(A\otimes d A)
\end{equation}
where $A$ is a vector field of the form (\ref{plane wave}).
\end{definition}

We list below the main properties of generalised pp-waves. Note that here and
further on we denote by $\{\!\nabla\!\}$ the covariant derivative
with respect to the Levi-Civita connection which should not be
confused with the full covariant derivative $\nabla$ incorporating
torsion.

The curvature of a generalised pp-wave is
\begin{equation}
\label{curvature of a generalised pp-space}
R=-\frac12(l\wedge\{\!\nabla\!\})\otimes(l\wedge\{\!\nabla\!\})f
+\frac14\mathrm{Re} \left( (h^2)''\,(l\wedge m)\otimes(l\wedge m)
\right).
\end{equation}
and the torsion of a generalised pp-wave is
\begin{equation}
\label{torsion of a generalised pp-space-general}
T=\mathrm{Re} \left(
(a\ l + b\ m) \otimes (l\wedge m)
\right),
\end{equation}
where
\[
a:=\frac12 h'(\varphi)\ k(\varphi), \quad
b:=\frac12 h'(\varphi)\ h(\varphi).
\]
Torsion can be written down even more explicitly in the following form
\begin{equation}\label{torsionVeryExplicit}
T=\sum_{j,k=1}^2t_{jk}\,m_j\otimes(l\wedge m_k)
+\sum_{j=1}^2t_j\,l\otimes(l\wedge m_j),
\end{equation} where
\[
t_{11}=-t_{22}=\frac12 \mathrm{Re }(b),\
t_{12}=t_{21}=-\frac12 \mathrm{Im }(b),\
t_1=\frac12 \mathrm{Re }(a), \ t_2=-\frac12 \mathrm{Im }(a),
\]  $m_1=\mathrm{Re } (m), m_2=\mathrm{Im }(m)$ and $a$ and $b$ at the same functions of the phase $\varphi$ appearing in equation (\ref{torsion of a generalised pp-space-general}).
\begin{remark}
From equation (\ref{torsionVeryExplicit}) we can clearly see that torsion has $4$ independent non-zero components.
\end{remark}
In the beginning of section \ref{Classical pp-waves} we introduced
the spinor field $\chi$ satisfying $\{\!\nabla\!\}\chi=0$. It becomes
clear that this spinor field also satisfies $\nabla\chi=0$.

\begin{lemma}\label{sameSpinorField}
The generalised pp-wave and the underlying classical pp-wave
admit the same nonvanishing parallel spinor field.
\end{lemma}
\begin{proof}
To see that $\nabla\chi=0$, we look at the only remaining torsion generated term from formula (\ref{spinor connection coefficient}) for the spinor connection coefficients $\Gamma^a{}_{\mu b}$, namely
\[
\nabla_\mu \chi^a= \{\!\nabla\!\}_\mu\chi^a + \frac14 \sigma^\alpha{}^{a\dot{c}} T_{\mu\alpha\beta}\sigma^\beta{}_{b\dot{c}} \chi^b.
\] In view of equation (\ref{torsion of a generalised pp-space-general}), it is sufficient to show that the term involving $(l\wedge m)$ contracted with the Pauli matrices gives zero. To see this, we rewrite the term in the following form
\begin{eqnarray}
\nonumber
\sigma^\alpha{}^{a\dot{c}} (l\wedge m)_{\alpha\beta}\sigma^\beta{}_{b\dot{c}} &=&
\frac12 (l\wedge m)_{\alpha\beta}( \sigma^\alpha{}^{a\dot{c}} \sigma^\beta{}_{b\dot{c}}-\sigma^\beta{}^{a\dot{c}} \sigma^\alpha{}_{b\dot{c}}) ,
\\
\nonumber
&=& (l\wedge m)_{\alpha\beta}\sigma^{\alpha\beta}{}^a{}_{b} = 0,
\end{eqnarray} which can be checked directly using our local coordinates (\ref{metric of a pp-wave}), (\ref{explicit l and a}) and the second order Pauli matrices (\ref{second order Pauli matrices for pp-metric}), i.e.
\[
-\sigma^{13}{}_{ab}-i\sigma^{23}{}_{ab}+\sigma^{31}{}_{ab}+i\sigma^{31}{}_{ab} = 0.
\] Hence, $\nabla \chi = 0$. \qed
\end{proof}
\begin{remark}
In view of Lemma \ref{sameSpinorField}, it is clear that both the generalised pp-wave and the underlying classical pp-wave admit the same nonvanishing parallel real null vector field $l$
and the same nonvanishing parallel complex 2-form (\ref{formula for F}), (\ref{decomposition of F}).
\end{remark}

\begin{lemma}\label{torsionPurelyTensorLemma}
The torsion (\ref{define torsion}) of a generalised pp-wave is purely tensor\footnote{Only the $T^{(1)}$, or `tensor torsion' irreducible piece of torsion is non-zero, see equations (\ref{torsion-pieces1}), (\ref{torsion-pieces2}).}, i.e.
\begin{equation*}
T^\alpha{}_{\alpha\gamma}=0, \qquad
\varepsilon_{\alpha\beta\gamma\delta}T^{\alpha\beta\gamma}=0.
\end{equation*}
\end{lemma}
\begin{proof}
The first equation $T^\alpha{}_{\alpha\gamma}=0$ follows directly from equation (\ref{torsion of a generalised pp-space-general}) and the fact that we have that
$l_\alpha l^\alpha=m_\alpha l^\alpha =m_\alpha m^\alpha=0$, see section \ref{Classical pp-waves}.
The second equation $\varepsilon_{\alpha\beta\gamma\delta}T^{\alpha\beta\gamma}=0$ follows from $*F=\pm\rmi
F$, see equations (\ref{formula for F}), (\ref{decomposition of F}). We then have
\[
\varepsilon_{\alpha\beta\gamma\delta} (l\wedge m)^{\beta\gamma} = Z (l\wedge m)_{\alpha\delta},
\] where $Z\in \mathbb{C}$ is some constant. Then using the formula for torsion (\ref{torsion of a generalised pp-space-general}) we have
\[
\varepsilon_{\alpha\beta\gamma\delta}T^{\alpha\beta\gamma} = \mathrm{Re} \left(
Z (a\ l + b\ m)^\alpha  (l\wedge m)_{\alpha\delta}
\right) = 0,
\] using the same argument as before, i.e. the fact that $l_\alpha l^\alpha=m_\alpha l^\alpha =m_\alpha m^\alpha=0$. \qed
\end{proof}

Examination of formula (\ref{curvature of a generalised pp-space}) for the curvature of a generalised pp-wave
reveals the following remarkable properties of generalised pp-waves:
\begin{itemize}
\item
The curvatures generated by the Levi-Civita connection and torsion
simply add up (compare formulae (\ref{curvature of a pp-space}) and
(\ref{curvature of a generalised pp-space})).
\item
The second term in the RHS of (\ref{curvature of a generalised
pp-space}) is purely Weyl. Consequently, the Ricci curvature of a
generalised pp-wave is completely determined by the pp-metric.
\item Clearly, generalised pp-waves have the same non-zero irreducible pieces of
curvature as classical pp-waves, namely symmetric trace--free Ricci and Weyl.
Using special local coordinates (\ref{metric of a pp-wave}),
(\ref{explicit l and a}), these can be expressed explicitly as
\begin{eqnarray*}
Ric &=& \frac12 (f_{11} + f_{22}) \, l\otimes l,
\\
\mathcal{W}&=& \sum_{j,k=1}^2 w_{jk} (l\wedge m_j)\otimes(l\wedge m_k),
\end{eqnarray*}
where $m_1=\mathrm{Re } (m), m_2=\mathrm{Im }(m)$, $f_{\alpha\beta} := \partial_{\alpha}\partial_{\beta} f$ and $w_{jk}$ are real scalars given by
\begin{eqnarray}
\nonumber
w_{11} = \frac14 [-f_{11} +f_{22} + \mathrm{Re } ((h^2)'')], & w_{22} = -w_{11},
\\
\nonumber
w_{12} = \pm \frac12 f_{12}-\frac14 \mathrm{Im } ((h^2)''), & w_{21} = w_{12}.
\end{eqnarray}
Compare these to the corresponding equations (\ref{explicit W and Ric1}) and (\ref{explicit W and Ric2}) for classical pp-waves.
\item
The curvature of a generalised pp-wave has all the usual symmetries
of curvature in the Riemannian case, that is,
\begin{eqnarray}
\label{symmetries of Riemannian curvature 1}
R_{\kappa\lambda\mu\nu}=R_{\mu\nu\kappa\lambda},
\\
\label{symmetries of Riemannian curvature 2}
\varepsilon^{\kappa\lambda\mu\nu}R_{\kappa\lambda\mu\nu}=0,
\\
\label{symmetries of Riemannian curvature 3}
R_{\kappa\lambda\mu\nu}=-R_{\lambda\kappa\mu\nu},
\\
\label{symmetries of Riemannian curvature 4}
R_{\kappa\lambda\mu\nu}=-R_{\kappa\lambda\nu\mu}.
\end{eqnarray}
Of course, (\ref{symmetries of Riemannian curvature 4}) is true for
any curvature whereas (\ref{symmetries of Riemannian curvature 3})
is a consequence of metric compatibility. Also, (\ref{symmetries of
Riemannian curvature 3}) follows from (\ref{symmetries of Riemannian
curvature 1}) and (\ref{symmetries of Riemannian curvature 4}).
\item
The second term in the RHS of (\ref{plane wave}) is pure gauge in
the sense that it does not affect curvature (\ref{curvature of a
generalised pp-space}). It does, however, affect torsion
(\ref{define torsion}).
\item
The Ricci curvature of a generalised pp-wave is zero if and only if
\begin{equation}
\label{ricci is zero} f_{11}+f_{22}=0
\end{equation}
and the Weyl curvature is zero if and only if
\begin{equation}
\label{weyl is zero} f_{11}-f_{22}=\mathrm{Re}\left((h^2)''\right),
\qquad f_{12}=\pm \frac12\mathrm{Im}\left((h^2)''\right).
\end{equation}
Here we use special local coordinates (\ref{metric of a pp-wave}),
(\ref{explicit l and a}) and denote
$f_{\alpha\beta}:=\partial_\alpha\partial_\beta f$.
\item
The curvature of a generalised pp-wave is zero if and only if we
have both (\ref{ricci is zero}) and (\ref{weyl is zero}). Clearly,
for any given function $h$ we can choose a function $f$ such that
$R=0$: this $f$ is a quadratic polynomial in $x^1$, $x^2$ with
coefficients depending on $x^3$. Thus, as a spin-off, we get a class
of examples of Weitzenb\"ock spaces ($T\ne0$, $R=0$).
\end{itemize}

%%%%%%%%%%%%%%%%%%%%%%%%%%%%%%%%%%%%%%%%%%
%      SPINOR FORMALISM FOR PP WAVES     %
%%%%%%%%%%%%%%%%%%%%%%%%%%%%%%%%%%%%%%%%%%

\subsection{Spinor formalism for generalised pp--waves}
\label{SpinorFormalismPPWaves}

In this section we provide the particular spinor formalism for
generalised pp-waves. For the pp-metric
(\ref{metric of a pp-wave}) we choose Pauli matrices
\begin{eqnarray}
\label{Pauli matrices for pp-metric}
\sigma^0{}_{a\dot b}= \left(
\begin{array}{cc}
1&0\\
0&-f
\end{array}
\right), \ \sigma^1{}_{a\dot b}= \left(
\begin{array}{cc}
0&1\\
1&0
\end{array}
\right), \nonumber
\\
\sigma^2{}_{a\dot b}= \left(
\begin{array}{cc}
0&\mp\rmi\\
\pm\rmi&0
\end{array}
\right),
\sigma^3{}_{a\dot b}= \left(
\begin{array}{cc}
0&0\\
0&2
\end{array}
\right).
\end{eqnarray}
Our two choices of Pauli matrices differ by orientation. When
dealing with a classical pp-wave the choice of orientation of Pauli
matrices does not really matter, but for a generalised pp-wave it
is convenient to choose orientation of Pauli matrices in agreement
with the sign in (\ref{polarized Maxwell equation}) and
(\ref{complex curvature is polarized})
as this simplifies the resulting formulae.

\begin{remark}
In the case $f=0$, formulae (\ref{Pauli matrices for pp-metric}) do
not turn into the Minkowski space Pauli matrices, since we write the metric in the form
(\ref{metric of a pp-wave}). This is a matter of convenience in
calculations.
\end{remark}

\begin{remark}
Note that we could have chosen a different set of Pauli matrices $\sigma^\alpha{}_{a\dot
b}$ in (\ref{Pauli matrices for pp-metric}), namely with the opposite sign in every Pauli matrix,
as they are a basis in the real
vector space of Hermitian $2\times2$ matrices $\sigma_{a\dot b}$ satisfying
$\sigma^\alpha{}_{a\dot b}\sigma^{\beta c\dot b}
+\sigma^\beta{}_{a\dot b}\sigma^{\alpha c\dot b}
=2g^{\alpha\beta}\delta_a{}^c$, defined uniquely up to a Lorentz transformation. See appendix \ref{AppendixSpinorFormalism} for more
on our chosen spinor formalism.
\end{remark}
Now we want to describe the spinor connection coefficients $\Gamma^a{}_{\mu b}$, see
formula (\ref{spinor connection coefficient}) in appendix
\ref{AppendixSpinorFormalism}.
For a generalised pp-wave, the non-zero coefficients of $ \Gamma^a{}_{\mu b}$ are
\begin{equation*}
\Gamma^1{}_{12}=\frac12hh', \qquad \Gamma^1{}_{22}=\mp \frac\rmi
2hh', \qquad \Gamma^1{}_{32}= \frac12\left(\frac{\partial
f}{\partial x^1} \pm\rmi\frac{\partial f}{\partial x^2}\right)
-\frac12kh'.
\end{equation*}
Here we use special local coordinates (\ref{metric of a pp-wave}),
(\ref{explicit l and a}) and Pauli matrices (\ref{Pauli matrices for
pp-metric}). Note that with our choice of Pauli matrices the signs
in formulae (\ref{Pauli matrices for pp-metric}) and
(\ref{polarization of second order Pauli matrices}) agree.

Since by Lemma \ref{torsionPurelyTensorLemma} the torsion of a generalised pp-wave is purely tensor, the massless Dirac equation (\ref{Weyl1}), (\ref{Weyl2}) (also called Weyl's equation) takes the form
\begin{equation}
\label{Weyl equation 1} \sigma^\mu{}_{a\dot b}\nabla_\mu\,\xi^a=0,
\end{equation}
or equivalently
\begin{equation}
\label{Weyl equation 2} \sigma^\mu{}_{a\dot
b}\{\!\nabla\!\}_\mu\,\xi^a=0,
\end{equation} see appendix \ref{AppendixWeyl} for more on the massless Dirac equation.

\begin{remark}
In view of equations (\ref{Weyl equation 1}) and (\ref{Weyl equation 2}),
it is easy to see that $\chi F(\varphi)$ is a solution of the massless Dirac
equation. Here $F$ is an arbitrary function of the phase
(\ref{phase}) and $\chi$ is the parallel spinor introduced in
section \ref{Classical pp-waves}.

We also provide the explicit formulae for the `second order Pauli
matrices' (\ref{second order Pauli matrices}) for the pp-metric
(\ref{metric of a pp-wave}).
\begin{eqnarray}
\sigma^{01}{}_{ab}= \left(
\begin{array}{cc}
1&0\\
0&f
\end{array}
\right),&
\sigma^{02}{}_{ab}= \left(
\begin{array}{cc}
\mp i&0\\
0&\pm if
\end{array}
\right),&
\sigma^{03}{}_{ab}= \left(
\begin{array}{cc}
0&1\\
1&0
\end{array}
\right),  \nonumber \\
\label{second order Pauli matrices for pp-metric}
\sigma^{12}{}_{ab}= \left(
\begin{array}{cc}
0&\mp i\\
\mp i&0
\end{array}
\right),&
\sigma^{13}{}_{ab}= \left(
\begin{array}{cc}
0&0\\
0&2
\end{array}
\right),&
\sigma^{23}{}_{ab}= \left(
\begin{array}{cc}
0&0\\
0&\pm 2i
\end{array}
\right).
\end{eqnarray}
\end{remark}
Note that as the second order Pauli
matrices $\sigma^{\alpha\beta}{}_{ab}$ are antisymmetric over the
tensor indices, i.e.
$\sigma^{\alpha\beta}{}_{ab}=-\sigma^{\beta\alpha}{}_{ab}$, we only
give the independent non-zero terms.

\section{Physical interpretation of generalised pp-waves}\label{interpretation}
It was shown in our previous work \cite{pasicvassiliev} that using the generalised
pp-waves described in section \ref{PP-waves with torsion} we can construct new vacuum solutions of
quadratic metric--affine gravity. The main result of \cite{pasicvassiliev} is the following
\begin{theorem}\label{main theorem}
Generalised pp-waves of parallel Ricci curvature are solutions of
the system of equations (\ref{eulerlagrangemetric}),
(\ref{eulerlagrangeconnection}).
\end{theorem}

The observation that one can construct vacuum solutions of
quadratic metric-affine gravity in terms of pp-waves is
a recent development. The fact that classical pp-spaces of parallel
Ricci curvature are solutions was first pointed out in
\cite{garda,poland,annalen}.

\begin{remark}
There is a slight error in our calculations of the explicit form of the field equations from \cite{pasicvassiliev}, which in no way influences the main result.
The error was noticed in producing \cite{skopje}, where the generalised version of the explicit field equations can be found, see Appendix \ref{AppendixCorrection} and for more details.
\end{remark}

In this section we attempt to give a
physical interpretation of the vacuum solutions of field
equations (\ref{eulerlagrangemetric}),
(\ref{eulerlagrangeconnection}) obtained in \cite{pasicvassiliev}, the notation and (partially) exposition from which we follow here. This topic was also explored very briefly and without
explicit calculations in our research review paper \cite{balkanica}. As noted in \cite{pasicvassiliev}, the following two classes of Riemannian spacetimes are solutions of our field equations:
\begin{itemize}
\item
Einstein spaces ($Ric=\Lambda g$), and
\item
classical pp-spaces of parallel Ricci curvature.
\end{itemize}
In general relativity, Einstein spaces are an accepted mathematical
model for vacuum. However, classical pp-spaces of parallel Ricci
curvature do not have an obvious physical interpretation. This
section gives an attempt at understanding whether our newly
constructed spacetimes are of mathematical or physical significance.

Our analysis of vacuum solutions of quadratic metric--affine gravity
shows that classical pp-spaces of
parallel Ricci curvature should not be viewed on their own. They are
a particular (degenerate) representative of a wider class of
solutions, namely, generalised pp-spaces of parallel Ricci
curvature. Indeed, the curvature of a generalised pp-space is a
sum of two curvatures: the curvature
\begin{equation}
\label{curvature of the underlying classical pp-space}
-\frac12(l\wedge\{\!\nabla\!\})\otimes(l\wedge\{\!\nabla\!\})f
\end{equation}
of the underlying classical pp-space and the curvature
\begin{equation}
\label{curvature generated by a torsion wave} \frac14\mathrm{Re}
\left( (h^2)''\,(l\wedge m)\otimes(l\wedge m) \right)
\end{equation}
generated by a torsion wave traveling over this classical pp-space.
Our torsion and
corresponding curvature (\ref{curvature generated by a torsion
wave}) are waves traveling at speed of light because $h$~and~$k$
are functions of the phase $\varphi$ which plays the role of a null
coordinate, $g^{\mu\nu}\nabla_\mu\varphi\,\nabla_\nu\varphi=0$.
The underlying classical pp-space of parallel
Ricci curvature can now be viewed as the `gravitational imprint'
created by a wave of some massless matter field. Such a situation
occurs in Einstein--Maxwell theory\footnote{A classical model describing the
interaction of gravitational and electromagnetic fields} and
Einstein--Weyl theory\footnote{A classical
model describing the interaction of gravitational and massless neutrino fields}.
The difference with our model is
that Einstein--Maxwell and Einstein--Weyl theories contain the
gravitational constant which dictates a particular relationship
between the strengths of the fields in question, whereas our model
is conformally invariant and the amplitudes of the two curvatures
(\ref{curvature of the underlying classical
pp-space})~and~(\ref{curvature generated by a torsion wave}) are
totally independent.

In the remainder of this
subsection we outline an argument in favour of
interpreting our torsion wave (\ref{define torsion}), (\ref{plane
wave}) as a mathematical model for some massless particle.

We base our interpretation on the analysis of the curvature
(\ref{curvature generated by a torsion wave}) generated by our
torsion wave. Examination of formula (\ref{curvature generated by a
torsion wave}) indicates that it is more convenient to deal with the
complexified curvature
\begin{equation}
\label{complexified curvature generated by a torsion wave}
\mathfrak{R}:=r\,(l\wedge m)\otimes(l\wedge m)
\end{equation}
where $r:=\frac14(h^2)''$ (this $r$ is a function of the phase
$\varphi$); note also that complexification is in line with the
traditions of quantum mechanics. Our complex curvature is polarized,
\begin{equation}
\label{complex curvature is polarized}
{}^*\mathfrak{R}=\mathfrak{R}^*=\pm\rmi\mathfrak{R}\,,
\end{equation}
and purely Weyl, hence it is equivalent to a (symmetric) rank 4
spinor~$\omega$. The relationship between $\mathfrak{R}$ and
$\omega$ is given by the formula
\begin{equation}
\label{spinor representation of curvature}
\mathfrak{R}_{\alpha\beta\gamma\delta} =\sigma_{\alpha\beta
ab}\,\omega^{abcd}\,\sigma_{\gamma\delta cd}
\end{equation}
where the $\sigma_{\alpha\beta}$ are `second order Pauli matrices'
(\ref{second order Pauli matrices}). Resolving (\ref{spinor
representation of curvature}) with respect to~$\omega$ we get, in
view of (\ref{formula for F}),
(\ref{decomposition of F}), (\ref{complexified curvature generated by a torsion wave}),
\begin{equation}
\label{formula for omega} \omega=\xi\otimes\xi\otimes\xi\otimes\xi
\end{equation}
where
\begin{equation}
\label{formula for xi} \xi:=r^{1/4}\,\chi
\end{equation}
and $\chi$ is the spinor field introduced in the beginning of
section \ref{Classical pp-waves}.

Formula (\ref{formula for omega}) shows that our rank 4 spinor
$\omega$ has additional algebraic structure: it is the 4th tensor
power of a rank 1 spinor $\xi$. Consequently, the complexified
curvature generated by our torsion wave is completely determined by
the rank 1 spinor field $\xi$.

We claim that the spinor field (\ref{formula for xi}) satisfies
the massless Dirac equation, see (\ref{Weyl equation 1}) or (\ref{Weyl equation
2}). Indeed, as $\chi$ is parallel checking that $\xi$ satisfies the
massless Dirac equation reduces to checking that
$\,(r^{1/4})'\,\sigma^\mu{}_{a\dot b}\,l_\mu\,\chi^a=0\,$. The
latter is established by direct substitution of the explicit formula (\ref{formula for l})
for $l$.

\subsection{Einstein--Weyl field equations}\label{EinsteinWeylFieldEquations}
In this section we aim to provide a reminder of
Einstein--Weyl theory and the field equations arising from this
classical model describing the interaction of gravitational and
massless neutrino fields, then to provide pp-wave type solutions within this
model, provide the previously known solutions of this type
and, finally, to compare them to the pp-wave type solutions of
our conformally invariant metric--affine model of gravity from \cite{pasicvassiliev}.

In Einstein--Weyl theory the action is given by
\begin{equation}\label{EWaction}
S_{EW}:=2i\int \left( \xi^a\,\sigma^\mu{}_{a\dot
b}\,(\{\!\nabla\!\}_\mu\overline\xi^{\dot b}) \ -\
(\{\!\nabla\!\}_\mu\xi^a)\,\sigma^\mu{}_{a\dot
b}\,\overline\xi^{\dot b} \right) + k \int \mathcal{R},
\end{equation} where the constant $k$
can be chosen so that the non-relativistic limit yields the usual
form of Newton's gravity law. According to Brill and Wheeler
\cite{wheeler}, $\displaystyle
k = \frac{c^4}{16\pi G},$
where $G$ is the gravitational constant.
\begin{remark}
Note that in Einstein--Weyl theory the connection is assumed to be
Levi-Civita, so we only vary the action (\ref{EWaction}) with
respect to the metric and the spinor.
\end{remark} We obtain the well known Einstein--Weyl field equations
\begin{equation}\label{firstEL}
\frac{\delta S_{EW}}{\delta g}=0,
\end{equation}
\begin{equation}\label{secondEL}
\frac{\delta S_{EW}}{\delta \xi}=0.
\end{equation}

The first term of the action $S$ depends on the spinor $\xi$ and the
metric $g$ while the second depends on the metric $g$ only. Hence
the formal variation of the action (\ref{EWaction}) with respect to
the spinor just yields the massless Dirac equation, see appendix
\ref{AppendixWeyl}.

The variation with respect to the metric is somewhat more
complicated as both terms of the action (\ref{EWaction}) depend on
the metric $g$. The variation of the first term of the action with
respect to the metric yields the energy momentum tensor of the Weyl action
(\ref{neutrinoAction}), i.e.

\begin{eqnarray}
E^{\mu\nu}&=&\frac{i}{2}\left[ \sigma^{\nu}{}_{a\dot b}
\left(\overline\xi^{\dot
b}\{\!\nabla\!\}^\mu\xi^a-\xi^a\{\!\nabla\!\}^\mu\overline\xi^{\dot
b}\right) +\sigma^{\mu}{}_{a\dot b} \left(\overline\xi^{\dot
b}\{\!\nabla\!\}^\nu\xi^a-\xi^a\{\!\nabla\!\}^\nu\overline\xi^{\dot
b}\right) \right] \nonumber
\\ \label{EMT}
&+&i\left( \xi^a\,\sigma^\eta{}_{a\dot
b}\,(\{\!\nabla\!\}_\eta\overline\xi^{\dot b})g^{\mu\nu} -
(\{\!\nabla\!\}_\eta\xi^a)\,\sigma^\eta{}_{a\dot
b}\,\overline\xi^{\dot b}g^{\mu\nu}\right).
\end{eqnarray} Note that the energy momentum
tensor is not a priori trace-free. Please see appendix \ref{EnergyMomentumTensorAppendix} for the detailed derivation of formula (\ref{EMT}).

Variation with respect to the metric of the Einstein--Hilbert term
of the action yields
\[
\delta\int \mathcal{R} = -
\int(Ric^{\mu\nu}-\frac12\mathcal{R}g^{\mu\nu})\delta g_{\mu\nu},
\] which is a straightforward calculation, see e.g. Landau and
Lifshitz \cite{LL2}. Hence we get the explicit representation of the Einstein--Weyl field
equations (\ref{firstEL}), (\ref{secondEL}):
\begin{eqnarray}
\frac{i}{2}\!\left[ \sigma^{\nu}{}_{a\dot b}
\left(\overline\xi^{\dot b}
\{\!\nabla\!\}^\mu\xi^a-\xi^a\{\!\nabla\!\}^\mu\overline\xi^{\dot b}
\right) \!+\!\sigma^{\mu}{}_{a\dot b} \left(\overline\xi^{\dot b}
\{\!\nabla\!\}^\nu\xi^a-\xi^a\{\!\nabla\!\}^\nu\overline\xi^{\dot b}
\right) \right] & & \nonumber
\\
\label{EWexplicit1} +i\left( \xi^a\,\sigma^\eta{}_{a\dot
b}\,(\{\!\nabla\!\}_\eta\overline\xi^{\dot b})g^{\mu\nu} -
(\{\!\nabla\!\}_\eta\xi^a)\,\sigma^\eta{}_{a\dot
b}\,\overline\xi^{\dot b}g^{\mu\nu}\right) & & \nonumber
\\
- k Ric^{\mu\nu}+\frac k2\mathcal{R}g^{\mu\nu}&=&0,
\\
\label{EWexplicit2} \sigma^\mu{}_{a\dot b}\{\! \nabla\! \}_\mu\,\xi^a
&=&0.
\end{eqnarray}
\begin{remark}\label{remarkEMT}
When the equation (\ref{EWexplicit2}) is satisfied, we have
that the energy--momentum tensor (\ref{EMT}) is trace free
and the second line of (\ref{EWexplicit1}) vanishes, see e.g. end of section 2 of
Griffiths and Newing \cite{griffiths1970_2}.
\end{remark}

\subsubsection{Known solutions of Einstein--Weyl theory}\label{EWComparisonExisting}

In one of the early works on this subject, Griffiths and Newing \cite{griffiths1970_1} show how the
solutions of Einstein--Weyl equations can be constructed and present five examples of solutions and a later work by the same authors \cite{griffiths1970_2} presents a more general solution of Kundt's class.
Audretsch and Graf \cite{audretsch3} derive a differential equation representing radiation solutions of the general relativistic Weyl's equation and study the corresponding energy-momentum tensor and they present an exact solution of Einstein--Weyl equations in the form of pp-waves. Audretsch  \cite{audretsch2} continues to study the asymptotic behaviour of the neutrino energy-momentum
tensor in curved space-time with the sole aid of generally covariant assumptions about the nature of the Weyl field and the author shows that these Weyl fields behave asymptotically like neutrino radiation.

Griffiths \cite{griffiths1972_1} expanded on his previous work and this paper is of particular interest to us as in section 5 of \cite{griffiths1972_1} the author presents solutions whose metric is the pp-wave metric (\ref{metric of a pp-wave}) and the author presents a condition on the function $f$ from the pp-metric (\ref{metric of a pp-wave}).
Griffiths \cite{griffiths1972_2} identifies a class of neutrino fields with zero energy momentum tensor and stipulates that these spacetimes may also be interpreted as describing gravitational waves. Collinson and Morris \cite{collinsonmorris} showed that these could be either pp-waves or Robinson-Trautman type N solutions presented in \cite{griffiths1970_1}. Subsequently these were called `\emph{ghost neutrinos}' by Davis and Ray in \cite{davisray}.

Kuchowicz and \. Zebrowski \cite{kuchowitz} expand on the work on ghost neutrinos trying to resolve this anomaly by considering non-zero torsion in the framework of Einstein--Cartan theory.
Griffiths \cite{griffiths1980} also considers the possibility of non-zero torsion and in a more general work \cite{griffiths1981} he showed that neutrino fields in Einstein-Cartan theory must have metrics that belong to the family of solutions of Kundt's class, which include the pp-waves. Singh and Griffiths \cite{singhgriffiths} corrected several mistakes from \cite{griffiths1981} and showed that neutrino fields in Einstein--Cartan theory also include the Robinson--Trautman type N solutions and that any solution of the Einstein--Weyl equations in general relativity has a corresponding solution in Einstein--Cartan theory. Thus pp-wave type solutions of Einstein--Weyl equations have corresponding solutions in Einstein--Cartan theory. This paper was one of the main inspirations behind the result in section \ref{ppwaveEW}.

\subsection{PP--wave type solutions of Einstein--Weyl theory}\label{ppwaveEW}
The aim of this section is to point out the fact that the
nonlinear system of equations (\ref{EWexplicit1}),
(\ref{EWexplicit2}) has solutions in the form of pp-waves.
Throughout this section we use the set of local coordinates
(\ref{metric of a pp-wave}), (\ref{explicit l and a}) and Pauli
matrices (\ref{Pauli matrices for pp-metric}). We now present a
class of explicit solutions of (\ref{EWexplicit1}),
(\ref{EWexplicit2}) where the metric $g$ is in the form of a
pp-metric and the spinor $\xi$ as in (\ref{formula for xi}). As
shown in the section \ref{interpretation}, the
spinor (\ref{formula for xi}) satisfies the equation
(\ref{EWexplicit2}). In the setting of a pp-space scalar curvature
vanishes and as the spinor $\chi$ appearing in formula (\ref{formula
for xi}) is parallel, in view of Remark \ref{remarkEMT} equation (\ref{EWexplicit1}) becomes
\begin{eqnarray}\nonumber
\frac{i}{2} \sigma^{\nu}{}_{a\dot b} \left(\overline\xi^{\dot
b}\{\!\nabla\!\}^\mu\xi^a-\xi^a\{\!\nabla\!\}^\mu\overline\xi^{\dot
b}\right) +\frac{i}{2}\sigma^{\mu}{}_{a\dot b}
\left(\overline\xi^{\dot
b}\{\!\nabla\!\}^\nu\xi^a\right.&-&\left.\xi^a\{\!\nabla\!\}^\nu\overline\xi^{\dot
b}\right) \\ \label{EW1evenMoreexplicit} &-& k Ric^{\mu\nu}  =  0.
\end{eqnarray}

We now need to determine under what condition the equation
(\ref{EW1evenMoreexplicit}) is satisfied. In our local coordinates,
we have
\begin{equation}\label{RicEW}
Ric = \left(\frac12 \frac{\partial^2 f}{\partial (x^1)^2} + \frac12
\frac{\partial^2 f}{\partial (x^2)^2}\right) (l\otimes l).
\end{equation}
Substituting formulae (\ref{formula for xi}), (\ref{RicEW}) into
equation (\ref{EW1evenMoreexplicit}), and using the fact that the
spinor $\chi$ is parallel, we obtain the equality
\[
i (\sigma^{\nu}{}_{a\dot b}l^\mu +\sigma^{\mu}{}_{a\dot b}l^\nu)\!\!
\left((r^{1/4})'\ \overline{r^{1/4}}-r^{1/4}\
(\overline{r^{1/4}})'\right)\chi^a\overline{\chi}^{\dot b}= \frac12
k l^\mu l^\nu \!\left(\frac{\partial^2 f}{\partial (x^1)^2} +
\frac{\partial^2 f}{\partial (x^2)^2}\right).
\]
Since we know that  $\sigma^\mu{}_{a\dot b}\chi^a\overline{\chi}^{\dot b}=l^\mu$,
we obtain the condition for a pp-wave type solution of the Einstein--Weyl model
\begin{equation}\label{condition on s}
\frac12 \frac{\partial^2 f}{\partial (x^1)^2} + \frac12
\frac{\partial^2 f}{\partial (x^2)^2} = \frac{i}{k}\left((r^{1/4})'\
\overline{r^{1/4}}-r^{1/4}\ (\overline{r^{1/4}})'\right).
\end{equation}
Thus, the complex valued function $r$ of one real variable $x^3$ can
be chosen arbitrarily and it uniquely determines the RHS of
(\ref{condition on s}). From (\ref{condition on s}) one recovers the
pp-metric by solving Poisson's equation.

\subsection{Comparison of metric--affine and Einstein--Weyl solutions}
\label{comparisonSection}
To make our comparison clearer, let us compare these models in the
case of mono--chromatic solutions of both models using local
coordinates (\ref{metric of a pp-wave}), (\ref{explicit l and a})
and Pauli matrices (\ref{Pauli matrices for pp-metric}).

\subsubsection{Monochromatic metric--affine solutions}
In the case of the metric--affine model, from Theorem \ref{main
theorem} we know that generalised pp-waves of parallel Ricci
curvature are solutions of the equations
(\ref{eulerlagrangemetric}), (\ref{eulerlagrangeconnection}).
Whether we are viewing monochromatic solutions or not, the condition
on the solution of the model remains unchanged, namely Ricci
curvature (\ref{RicEW}) has to be parallel. In our special local
coordinates the condition of parallel Ricci curvature is written as
\begin{equation}\label{conditionOnSolutionsMAG}
\frac12 \frac{\partial^2 f}{\partial (x^1)^2} + \frac12
\frac{\partial^2 f}{\partial (x^2)^2} = C,
\end{equation} where $C$ is an arbitrary real constant.
However, the construction of torsion simplifies in the monochromatic
case. Namely, we can choose the function $h$ of the phase
(\ref{phase}) so that the plane wave (\ref{plane wave}) becomes
\[
A = \frac{ic^2}{2a} \ e^{2i(ax^3 +b)}\ m,
\] where $a,b,c \in \mathbb{R}$, $a\ne0$.
Torsion (\ref{define torsion}) then takes the form
\[
T  =  -\frac{c^4}{4a}\textrm{Re}\left(ie^{4i(ax^3+b)}
m\otimes(l\wedge m)\right).
\]
Hence the complexified curvature (\ref{complexified curvature
generated by a torsion wave}) generated by the torsion wave becomes
\[
\mathfrak{R} = c^4 e^{4i(ax^3 +b)}(l\wedge m) \otimes (l\wedge m),
\]
and $r$ from (\ref{complexified curvature generated by a torsion
wave}) becomes
\[
r = c^4 e^{4i(ax^3 +b)}.
\]
The spinor $\xi$ from (\ref{formula for xi}) is
explicitly given by
\begin{equation}
\label{dima1}
\xi = c \left( \begin{array}{c}1 \\ 0\end{array}\right)
e^{i(ax^3+b)}.
\end{equation}

\subsubsection{Monochromatic Einstein--Weyl solutions}
Let us now look for monochromatic solutions in Einstein--Weyl theory.
We take the spinor field as in formula (\ref{dima1}) in which case
condition (\ref{condition on s}) simplifies to
\begin{equation}\label{conditionOnSolutionsEW}
\frac12 \frac{\partial^2 f}{\partial (x^1)^2} + \frac12
\frac{\partial^2 f}{\partial (x^2)^2} = - \frac{2ac^2}{k}.
\end{equation}

\subsubsection{Comparison of monochromatic metric--affine and Einstein--Weyl solutions}
The main difference between the two models is that in the
metric--affine model our generalised pp-wave solutions  have
parallel Ricci curvature, whereas in the Einstein--Weyl model the
pp-wave type solutions do not necessarily have parallel Ricci
curvature. However, when we look at monochromatic pp-wave type
solutions in the Einstein--Weyl model their Ricci curvature also
becomes parallel. The only remaining difference is in the right-hand
sides of equations (\ref{conditionOnSolutionsMAG}) and
(\ref{conditionOnSolutionsEW}): in (\ref{conditionOnSolutionsMAG})
the constant is arbitrary whereas in (\ref{conditionOnSolutionsEW})
the constant is expressed via the characteristics of the spinor wave
and the gravitational constant.

In other words, comparing equations (\ref{conditionOnSolutionsMAG})
and (\ref{conditionOnSolutionsEW}) we see that while in the
metric--affine case the Laplacian of $f$ can be \emph{any} constant,
in the Einstein--Weyl case it is required to be a \emph{particular}
constant. This should not be surprising as our metric--affine model
is conformally invariant,
while the Einstein--Weyl model is not.

We also want to clarify that $f$ and the quantities $a,b,c$ appearing in this
section \ref{comparisonSection} are generally arbitrary functions of the null coordinate $x^3$.
As such, if these quantities are non-zero only for a
short finite interval of $x^3$, the solutions represent spinors, curvature
and torsion components which propagate at the speed of light.

Hence we can conclude that
generalised pp-waves of parallel Ricci curvature are very
similar to pp-type solutions of the Einstein--Weyl model, which is a
classical model describing the interaction of massless neutrino and
gravitational fields. Therefore we suggest that
\begin{center}
\emph{Generalised pp-waves of parallel Ricci curvature represent a
metric--affine model for the massless neutrino.}
\end{center}

\begin{acknowledgements}
The authors are very grateful to D Vassiliev, J B Griffiths and F W Hehl for helpful advice and to the Ministry of Education and Science of the Federation of Bosnia and Herzegovina, which
supported our research.
\end{acknowledgements}

\appendix
\section{Spinor Formalism}\label{AppendixSpinorFormalism}

This appendix provides the spinor formalism used throughout our work.
Unless otherwise stated, we work in a general metric
compatible spacetime with torsion. When introducing our spinor formalism, we were faced with the
problem that there doesn't seem to exist a uniform convention in the
existing literature on how to treat spinors. Optimally, we would
have wanted to achieve the following:

\begin{enumerate}
\item[(i)] consecutive raising and lowering of a spinor index does not
change the sign of a rank $1$ spinor;
\item[(ii)] the metric spinor $\epsilon^{ab}$ is the raised version
of $\epsilon_{ab}$ and vice versa;
\item[(iii)] the spinor inner product is invariant under raising and
lowering of indices, i.e. $\xi_a \eta^a = \xi^a \eta_a $.
\end{enumerate}

Unfortunately, it becomes clear that it is not possible to satisfy all
three desired properties, as shown in \cite{pirani}. This
inconsistency is related to the well known fact (see for example
Section 19 in \cite{LL4} or Section 3--5 in \cite{StreaterWightman}),
that a spinor does not have a particular sign -- for example, a
spatial rotation of the coordinate system by $2\pi$ leads to a
change of sign. Also see \cite{penroserindler} for more helpful
insight about the problem of choice of the spinor formalism, as well as e.g.
\cite{LL4,blagojevic,buchbinderandkuzenko,griffiths1981,griffiths3,pirani} for
insight to various approaches to spinor formalism.

We decided to define our spinor formalism in the following way.
We define the `metric spinor' as
\begin{equation}
\label{metric spinor} \epsilon_{ab}=\epsilon_{\dot a\dot b}=
\epsilon^{ab}=\epsilon^{\dot a\dot b}= \left(
\begin{array}{cc}
0&1\\
-1&0
\end{array}
\right)
\end{equation}
with the first index enumerating rows and the second enumerating
columns. We raise and lower spinor indices according to the formulae
\begin{equation}
\label{raising and lowering of spinor indices}
\xi^a=\epsilon^{ab}\xi_b, \qquad \xi_a=\epsilon_{ab}\xi^b, \qquad
\eta^{\dot a}=\epsilon^{\dot a\dot b}\eta_{\dot b}, \qquad
\eta_{\dot a}=\epsilon_{\dot a\dot b}\eta^{\dot b}.
\end{equation}

Our definition (\ref{metric spinor}), (\ref{raising and lowering of
spinor indices}) has the following advantages:
\begin{itemize}
\item
The spinor inner product is invariant under the operation of raising
and lowering of indices, i.e.
$(\epsilon_{ac}\xi^c)(\epsilon^{ad}\eta_d)=\xi^a\eta_a$.
\item
The `contravariant' and `covariant' metric spinors are
`raised' and `lowered' versions of each other, i.e.
$\epsilon^{ab}=\epsilon^{ac}\epsilon_{cd}\epsilon^{bd}$ and
$\epsilon_{ab}=\epsilon_{ac}\epsilon^{cd}\epsilon_{bd}$.
\end{itemize}
The disadvantage of our definition (\ref{metric spinor}),
(\ref{raising and lowering of spinor indices}) is that the
consecutive raising and lowering of a single spinor index leads to a
change of sign, i.e. $\epsilon_{ab}\epsilon^{bc}\xi_c=-\xi_a$. In
formulae where the sign is important we will be careful in
specifying our choice of sign; see, for example, (\ref{definition of
contravariant Pauli matrices}), (\ref{spinor connection
coefficient}). We in a sense intentionally `sacrificed' this
property in order to guarantee that the other two properties, which
in our view have greater physical significance, are satisfied.

Let $\mathfrak{v}$ be the real vector space of Hermitian $2\times2$
matrices $\sigma_{a\dot b}$. Pauli matrices $\sigma^\alpha{}_{a\dot
b}$, $\alpha=0,1,2,3$, are a basis in $\mathfrak{v}$ satisfying
$\sigma^\alpha{}_{a\dot b}\sigma^{\beta c\dot b}
+\sigma^\beta{}_{a\dot b}\sigma^{\alpha c\dot b}
=2g^{\alpha\beta}\delta_a{}^c$ where
\begin{equation}
\label{definition of contravariant Pauli matrices} \sigma^{\alpha
a\dot b}:= \epsilon^{ac}\sigma^\alpha{}_{c\dot d}\epsilon^{\dot
b\dot d}.
\end{equation}
At every point of the manifold $M$ Pauli matrices are defined
uniquely up to a Lorentz transformation.
Define
\begin{equation}
\label{second order Pauli matrices} \sigma_{\alpha\beta ac}:=\frac12
\bigl( \sigma_{\alpha a\dot b}\epsilon^{\dot b\dot d}\sigma_{\beta
c\dot d} - \sigma_{\beta a\dot b}\epsilon^{\dot b\dot
d}\sigma_{\alpha c\dot d} \bigr)\,.
\end{equation}
These `second order Pauli matrices' are polarized, i.e.
\begin{equation}
\label{polarization of second order Pauli matrices}
*\sigma=\pm i\sigma
\end{equation}
depending on the orientation of `basic' Pauli matrices
$\sigma^\alpha{}_{a\dot b}$, $\alpha=0,1,2,3$.

We define the covariant derivatives of spinor fields as
\[
\nabla_\mu\xi^a=\partial_\mu\xi^a+\Gamma^a{}_{\mu b}\xi^b, \qquad
\nabla_\mu\xi_a=\partial_\mu\xi_a-\Gamma^b{}_{\mu a}\xi_b,
\]
\[
\nabla_\mu\eta^{\dot a}=\partial_\mu\eta^{\dot a} +\bar\Gamma^{\dot
a}{}_{\mu\dot b}\eta^{\dot b}, \qquad \nabla_\mu\eta_{\dot
a}=\partial_\mu\eta_{\dot a} -\bar\Gamma^{\dot b}{}_{\mu\dot
a}\eta_{\dot b},
\]
where $\bar\Gamma^{\dot a}{}_{\mu\dot b}=\overline{\Gamma^a{}_{\mu
b}}$. The explicit formula for the spinor connection coefficients
$\Gamma^a{}_{\mu b}$ can be derived from the following two
conditions:
\begin{equation}
\label{condition 12} \nabla_\mu\epsilon^{ab}=0, \quad
\nabla_\mu j^\alpha=\sigma^\alpha{}_{a\dot
b}\nabla_\mu\zeta^{a\dot b},
\end{equation}
where $\zeta$ is an arbitrary rank 2 mixed spinor field and
$j^\alpha:=\sigma^\alpha{}_{a\dot b}\zeta^{a\dot b}$ is the
corresponding vector field (current). Conditions (\ref{condition
12}) give a system of linear algebraic equations
for $\mathrm{Re}\,\Gamma^a{}_{\mu b}$, $\mathrm{Im}\,\Gamma^a{}_{\mu
b}$ the unique solution of which is
\begin{equation}
\label{spinor connection coefficient} \Gamma^a{}_{\mu b}=\frac14
\sigma_\alpha{}^{a\dot c} \left(
\partial_\mu\sigma^\alpha{}_{b\dot c}
+\Gamma^\alpha{}_{\mu\beta}\sigma^{\beta}{}_{b\dot c} \right).
\end{equation}
See section 3 of \cite{griffiths1970_1} for more
background on covariant differentiation of spinors.

\section{Massless Dirac Equation}\label{AppendixWeyl} The generally accepted
point of view \cite{griffiths1981,hehl habilitation,hehl
neutrino 1,hehl neutrino 2,hehl neutrino 3} is that a massless neutrino
field is a metric compatible spacetime with or without torsion,
described by the action
\begin{equation}\label{neutrinoAction}
S_\mathrm{neutrino}:=2i\int \Bigl( \xi^a\,\sigma^\mu{}_{a\dot
b}\,(\nabla_\mu\bar\xi^{\dot b}) \ -\
(\nabla_\mu\xi^a)\,\sigma^\mu{}_{a\dot b}\,\bar\xi^{\dot b} \Bigr),
\end{equation} see formula (11) of \cite{griffiths1981}.
We first vary the action (\ref{neutrinoAction}) with respect to the
spinor $\xi$, while keeping torsion and the metric \emph{fixed}. A
straightforward calculation produces the \emph{massless Dirac (or Weyl's)} equation
\begin{equation}\label{Weyl1}
\sigma^\mu{}_{a\dot b}\nabla_\mu\,\xi^a
-\frac12T^\eta{}_{\eta\mu}\sigma^\mu{}_{a\dot b}\,\xi^a=0,
\end{equation}
which can be equivalently rewritten as
\begin{equation}\label{Weyl2}
\sigma^\mu{}_{a\dot b}\{\!\nabla\!\}_\mu\,\xi^a \pm\frac\rmi 4
\varepsilon_{\alpha\beta\gamma\delta}T^{\alpha\beta\gamma}
\sigma^\delta{}_{a\dot b}\,\xi^a=0.
\end{equation}

\subsection{Energy momentum tensor}\label{EnergyMomentumTensorAppendix}

In this subsection we give the derivation of the energy momentum
tensor of the action $S_\mathrm{neutrino}$, where we vary the metric
keeping the spinor fixed. The covariant and contravariant metric change in the following way
\begin{equation}\label{VariationOfCovariantMetric}
g_{\alpha\beta}\mapsto g_{\alpha\beta} + \delta g_{\alpha\beta}, \qquad
g^{\alpha\beta}\mapsto g^{\alpha\beta} -g^{\alpha\alpha'}(\delta
g_{\alpha'\beta'})g^{\beta\beta'},
\end{equation} while the Pauli matrices transform in the
following way
\begin{equation}\label{SymmetricVariationOfCovariantandContravariantPauliMatrices}
\sigma_{\alpha}\mapsto \sigma_{\alpha} + \frac12\delta
g_{\alpha\beta} g^{\beta\gamma}\sigma_{\gamma}, \ \quad \ \sigma^\alpha \mapsto \sigma^\alpha - \frac12g^{\alpha\beta}(\delta
g_{\beta\gamma}) \sigma^{\gamma}.
\end{equation}
Formulae describe a `symmetric' variation of the Pauli matrices
caused by the (symmetric) variation of the (symmetric) metric.

\begin{remark}\label{EMTsimple}
We do most of the following calculations
under the assumption that the metric is the Minkow--ski metric
$g_{\mu\nu}=\textrm{diag} (1,-1,-1,-1)$ and that the
connection is Levi-Civita.
\end{remark}
Now we need to look at the $\delta \Gamma^{\alpha}{}_{\beta\gamma}$.
Using the definition of the Levi-Civita connection,
equation (\ref{VariationOfCovariantMetric}) and
metric compatibility ($\nabla g \equiv 0$),  we get that
the connection transforms as
\begin{equation}\label{VariationOfLeviCivitaConnection}
\delta \Gamma^{\kappa}{}_{\mu\nu} =
\frac12g^{\kappa\lambda}(\nabla_\mu \delta g_{\lambda\nu}
+\nabla_\nu \delta g_{\lambda\mu}-\nabla_\lambda \delta g_{\mu\nu}).
\end{equation}
\begin{lemma}
The variation of the covariant derivative of $\xi$ with respect to
the metric is
\begin{equation}\label{VariationOfCovariantDerivative}
\delta\nabla_\mu\xi^a =\frac18(\sigma_\alpha{}^{a\dot
d}\sigma^{\beta}{}_{c\dot d} -\sigma^\beta{}^{a\dot
d}\sigma_{\alpha}{}_{c\dot d})
\xi^c\delta\Gamma^\alpha{}_{\mu\beta}.
\end{equation}
\end{lemma}
\begin{proof}
Using equation (\ref{spinor connection coefficient}), the fact that $\xi$ does not contribute to the variation
and the assumptions in Remark \ref{EMTsimple}, we obtain
\[
4\delta\nabla_\mu\xi^a = \sigma_\alpha{}^{a\dot d}\left(
\partial_\mu(\delta\sigma^\alpha{}_{c\dot d})
+(\delta\Gamma^\alpha{}_{\mu\beta})\sigma^{\beta}{}_{c\dot d}
\right)\xi^c.
\]
Using equation
(\ref{VariationOfCovariantMetric}) and metric compatibility
we get that
\[
\partial_\mu(\delta\sigma^\alpha{}_{c\dot d})=
-\frac12 g^{\alpha\eta}\sigma_{\zeta}{}_{c\dot d}\ \delta
\Gamma^{\zeta}{}_{\mu\eta} -\frac12 \delta^\alpha{}_{\zeta}\
\sigma^{\xi}{}_{c\dot d}\ \delta\Gamma^{\zeta}{}_{\mu\xi}.
\]
Combining this with the formula for the variation of $\nabla \xi$,
we get the equivalent to equation (\ref{VariationOfCovariantDerivative})
\[
4\delta\nabla_\mu\xi^a = -\frac12\sigma^\beta{}^{a\dot d}
\sigma_{\alpha}{}_{c\dot d}\ \xi^c \delta
\Gamma^{\alpha}{}_{\mu\beta} -\frac12\sigma_\alpha{}^{a\dot d} \
\sigma^{\beta}{}_{c\dot d}\ \xi^c \delta\Gamma^{\alpha}{}_{\mu\beta}
+\sigma^{\beta}{}_{c\dot
d}\sigma_\alpha{}^{a\dot d}\xi^c \delta\Gamma^\alpha{}_{\mu\beta}.
\] \qed
\end{proof}
We now combine equations
(\ref{VariationOfLeviCivitaConnection}) and
(\ref{VariationOfCovariantDerivative}) to get
\[
\delta\nabla_\mu\xi^a =\frac1{16}(\sigma^\lambda{}^{a\dot
d}\sigma^{\beta}{}_{c\dot d} -\sigma^\beta{}^{a\dot
d}\sigma^{\lambda}{}_{c\dot d}) \xi^c\ (\partial_\mu \delta
g_{\lambda\beta} +\partial_\beta \delta g_{\lambda\mu}
-\partial_\lambda \delta g_{\mu\beta}).
\]
As the first derivative is symmetric over $\lambda,\beta$ and the
Pauli matrices are antisymmetric over these indices, we get
\[
(\sigma^\lambda{}^{a\dot d}\sigma^{\beta}{}_{c\dot d}
-\sigma^\beta{}^{a\dot d}\sigma^{\lambda}{}_{c\dot d})\partial_\mu
\delta g_{\lambda\beta}= -(\sigma^\lambda{}^{a\dot
d}\sigma^{\beta}{}_{c\dot d} -\sigma^\beta{}^{a\dot
d}\sigma^{\lambda}{}_{c\dot d})\partial_\mu \delta g_{\beta\lambda}
=0.
\] Hence,
\[
\delta\nabla_\mu\xi^a =\frac1{16}(\sigma^\lambda{}^{a\dot
d}\sigma^{\beta}{}_{c\dot d} -\sigma^\beta{}^{a\dot
d}\sigma^{\lambda}{}_{c\dot d}) \xi^c \partial_\beta \delta
g_{\lambda\mu}\ -\frac1{16}(\sigma^\beta{}^{a\dot
d}\sigma^{\lambda}{}_{c\dot d} -\sigma^\lambda{}^{a\dot
d}\sigma^{\beta}{}_{c\dot d})\xi^c\partial_\beta\delta
g_{\mu\lambda}.
\] So finally, we get the formula for the variation of the covariant derivative of $\xi$:

\begin{equation}\label{VariationOfCovariantDerivativeFullSpecialCase}
\delta\{\! \nabla\! \}_\mu\xi^a =\frac18\xi^c (\sigma^\alpha{}^{a\dot
d}\sigma^{\beta}{}_{c\dot d} -\sigma^\beta{}^{a\dot
d}\sigma^{\alpha}{}_{c\dot d})
\partial_\beta \delta g_{\mu\alpha}.
\end{equation}
\begin{lemma}
The energy momentum tensor of the action (\ref{neutrinoAction}) is equation (\ref{EMT}).
\end{lemma}
\begin{proof}
Varying the action (\ref{neutrinoAction}) with respect to the metric, we get
\begin{eqnarray*}
\delta S &=& 2i\delta\int\left(\xi^a\,\sigma^\eta{}_{a\dot
b}\,(\{\! \nabla\! \}_\eta\overline\xi^{\dot b}) \ -\
(\{\! \nabla\! \}_\eta\xi^a)\,\sigma^\eta{}_{a\dot b}\,\overline\xi^{\dot
b}\right) \sqrt{|\det g|}
\\
&=&2i\int\xi^a\,(\delta\sigma^\eta{}_{a\dot b})\,(\{\! \nabla\! \}_\eta\overline\xi^{\dot b})
+\xi^a\,\sigma^\eta{}_{a\dot
b}\,(\delta\{\! \nabla\! \}_\eta\overline\xi^{\dot b}) - (\delta
\{\! \nabla\! \}_\eta\xi^a)\,\sigma^\eta{}_{a\dot b}\,\overline\xi^{\dot b}
-(\{\! \nabla\! \}_\eta\xi^a)\,(\delta\sigma^\eta{}_{a\dot
b})\,\overline\xi^{\dot b}
\\
&&+\frac12\left(\xi^a\,\sigma^\eta{}_{a\dot
b}\,(\{\! \nabla\! \}_\eta\overline\xi^{\dot b}) \ -\
(\{\! \nabla\! \}_\eta\xi^a)\,\sigma^\eta{}_{a\dot b}\,\overline\xi^{\dot
b}\right)g^{\mu\nu}\delta g_{\mu\nu}
\end{eqnarray*}
and using equation
(\ref{SymmetricVariationOfCovariantandContravariantPauliMatrices}) we get
\begin{eqnarray*}
\delta S &=&
2i\int\left(-\frac14\xi^a\,g^{\eta\mu}\sigma^{\nu}{}_{a\dot b}\,
(\{\! \nabla\! \}_\eta\overline\xi^{\dot b})
-\frac14\xi^a\,g^{\eta\nu}\sigma^{\mu}{}_{a\dot b}\,
(\{\! \nabla\! \}_\eta\overline\xi^{\dot b})
+\frac14(\{\! \nabla\! \}_\eta\xi^a)\,g^{\eta\mu}\sigma^{\nu}{}_{a\dot
b}\,\overline\xi^{\dot b}\right.
\\
&+& \left.\frac14(\{\! \nabla\! \}_\eta\xi^a)\,g^{\eta\nu}\sigma^{\mu}{}_{a\dot
b}\,\overline\xi^{\dot b} +\frac12\xi^a\,\sigma^\eta{}_{a\dot
b}\,(\{\! \nabla\! \}_\eta\overline\xi^{\dot b})g^{\mu\nu} -
\frac12(\{\! \nabla\! \}_\eta\xi^a)\,\sigma^\eta{}_{a\dot
b}\,\overline\xi^{\dot b}g^{\mu\nu}\right) \delta g_{\mu\nu}
\\
&+&\xi^a\,\sigma^\eta{}_{a\dot
b}\,(\delta\{\! \nabla\! \}_\eta\overline\xi^{\dot b}) -(\delta
\{\! \nabla\! \}_\eta\xi^a)\,\sigma^\eta{}_{a\dot b}\,\overline\xi^{\dot b}.
\end{eqnarray*}
Now we look at the terms involving the variation of $\{\! \nabla\! \} \xi$ on
their own. Using equation
(\ref{VariationOfCovariantDerivativeFullSpecialCase}) we get
\begin{equation*}
I_1 = \frac{i}{4} \int \xi^a\,\sigma^\mu{}_{a\dot b}\,
\overline{\xi}^{\dot d}(\sigma^\nu{}^{c\dot b}\sigma^{\eta}{}_{c\dot
d} -\sigma^\eta{}^{c\dot b}\sigma^{\nu}{}_{c\dot d})\partial_\eta
\delta g_{\mu\nu}- \xi^c (\sigma^\nu{}^{a\dot
d}\sigma^{\eta}{}_{c\dot d} -\sigma^\eta{}^{a\dot
d}\sigma^{\nu}{}_{c\dot d})
\partial_\eta \delta g_{\mu\nu}
\,\sigma^\mu{}_{a\dot b}\,\overline\xi^{\dot b}.
\end{equation*}
Integrating by parts and using the simplifications from Remark \ref{EMTsimple}, we get
\[
I_1=
\frac{i}{4} \int \overline{\xi}^{\dot b}\{\! \nabla\! \}_\eta\xi^a
\, \left( -\sigma^\mu{}_{a\dot d}\sigma^\nu{}^{c\dot
d}\sigma^{\eta}{}_{c\dot b} +\sigma^\mu{}_{a\dot
d}\sigma^\eta{}^{c\dot d}\sigma^{\nu}{}_{c\dot b}+
\sigma^{\eta}{}_{a\dot d}\sigma^\nu{}^{c\dot d}\sigma^\mu{}_{c\dot
b} -\sigma^{\nu}{}_{a\dot d}\sigma^\eta{}^{c\dot
d}\sigma^\mu{}_{c\dot b}\right.
\]
\[
\left.-\sigma^\mu{}_{a\dot d}\sigma^\nu{}^{c\dot d}\sigma^{\eta}{}_{c\dot
b} +\sigma^\mu{}_{a\dot d}\sigma^\eta{}^{c\dot
d}\sigma^{\nu}{}_{c\dot b} +\sigma^{\eta}{}_{a\dot
d}\sigma^\nu{}^{c\dot d}\sigma^\mu{}_{c\dot b}
-\sigma^{\nu}{}_{a\dot d}\sigma^\eta{}^{c\dot d}\sigma^\mu{}_{c\dot
b} \right) \delta g_{\mu\nu}.
\]Since we have $\displaystyle
(\sigma^\mu{}_{a\dot d}\sigma^\eta{}^{c\dot d}\sigma^{\nu}{}_{c\dot
b} -\sigma^{\nu}{}_{a\dot d}\sigma^\eta{}^{c\dot
d}\sigma^\mu{}_{c\dot b}) \delta g_{\mu\nu} =0,
$ as it is a product of symmetric and antisymmetric tensors, as well as (after a lengthy but straightforward calculation)
\begin{equation*}\label{cubic term in sigma zero}
\sigma^{\eta}{}_{a\dot d}\sigma^\nu{}^{c\dot d}\sigma^\mu{}_{c\dot
b}+ \sigma^{\eta}{}_{a\dot d}\sigma^\mu{}^{c\dot
d}\sigma^\nu{}_{c\dot b} -\sigma^\mu{}_{a\dot d}\sigma^\nu{}^{c\dot
d}\sigma^{\eta}{}_{c\dot b} -\sigma^\nu{}_{a\dot
d}\sigma^\mu{}^{c\dot d}\sigma^{\eta}{}_{c\dot b} = 0,
\end{equation*}
we have shown that the terms involving $\delta \{\! \nabla\! \} \xi$ do
not contribute to the variation, i.e. $I_1 = 0.$
We now return to the variation of the whole action, which after
some simplification becomes
\[
\frac{\delta S}{\delta g} = \frac{i}2\int\left(\sigma^{\nu}{}_{a\dot
b} ((\{\! \nabla\! \}^\mu\xi^a)\overline\xi^{\dot
b}-\xi^a\{\! \nabla\! \}^\mu\overline\xi^{\dot b}) +\sigma^{\mu}{}_{a\dot b}
((\{\! \nabla\! \}^\nu\xi^a)\overline\xi^{\dot
b}-\xi^a\{\! \nabla\! \}^\nu\overline\xi^{\dot b})\right)\delta g_{\mu\nu}
\]
\[
+i\int \left( \xi^a\,\sigma^\eta{}_{a\dot
b}\,(\{\! \nabla\! \}_\eta\overline\xi^{\dot b})g^{\mu\nu} -
(\{\! \nabla\! \}_\eta\xi^a)\,\sigma^\eta{}_{a\dot b}\,\overline\xi^{\dot
b}g^{\mu\nu}\right) \delta g_{\mu\nu}.
\]
Finally we can conclude that the energy momentum tensor of the action (\ref{neutrinoAction}) is exactly equation (\ref{EMT}). \qed
\end{proof}

\renewcommand{\thesection}{C}
  % redefine the command that creates the equation no.
  \setcounter{section}{0}  % reset counter

\section{Correction of explicit form of the second field equation and future work}\label{AppendixCorrection}

In calculating the Bianchi identity for curvature in Appendix B from \cite{pasicvassiliev}, there was an arithmetic error, hence equation (B.11) from that work should read
\[
\nabla_\eta Ric^\eta{}_\lambda =
-\frac12 Ric^\eta{}_\xi T_\eta{}^\xi{}_{\lambda} -\frac12 \mathcal{W}^{\eta\zeta}{}_{\lambda\xi}(T_\eta{}^\xi{}_{\zeta}-T_\zeta{}^\xi{}_{\eta}),
\] and from here equation (B.12) should read
{\small
\begin{eqnarray*}
\nabla_\eta \mathcal{W}^\eta{}_{\mu\lambda\kappa} &=& \mathcal{W}^\eta{}_{\mu\kappa\xi}(T_\eta{}^\xi{}_{\lambda}-T_\lambda{}^\xi{}_{\eta} )
+ \mathcal{W}^\eta{}_{\mu\lambda\xi}(T_\kappa{}^\xi{}_{\eta}-T_\eta{}^\xi{}_{\kappa}) \\
&+&\frac14(T_\zeta{}^\xi{}_{\eta}-T_\eta{}^\xi{}_{\zeta})(g_{\mu\lambda}\mathcal{W}^{\eta\zeta}{}_{\kappa\xi}-g_{\mu\kappa}\mathcal{W}^{\eta\zeta}{}_{\lambda\xi})
+\frac14Ric^\eta{}_\xi
(g_{\mu\lambda}T_\eta{}^\xi{}_{\kappa}-g_{\mu\kappa}T_\eta{}^\xi{}_{\lambda})
\\
&+&\frac12\left[\nabla_\lambda Ric_{\mu\kappa}-
\nabla_\kappa Ric_{\mu\lambda} + Ric^\eta{}_\kappa(T_{\lambda\eta\mu}-T_{\eta\lambda\mu})
+Ric^\eta{}_\lambda(T_{\eta\kappa\mu}-T_{\kappa\eta\mu})\right]
\end{eqnarray*}} Consequently, when these two are used in calculating the explicit form of the field equation (\ref{eulerlagrangeconnection})
this produces a different result to the one presented in \cite{pasicvassiliev}. Namely, equation (\ref{eulerlagrangeconnection}) in its explicit form (i.e. equation (27) from \cite{pasicvassiliev}) should read \begin{eqnarray*}
d_6 \nabla_\lambda Ric_{\kappa\mu}  & - & d_7 \nabla_\kappa Ric_{\lambda\mu}
\\
&+&d_6\left(
Ric^\eta{}_\kappa\left(T_{\eta\mu\lambda}-T_{\lambda\mu\eta}\right)
+\frac12 g_{\mu\lambda} \mathcal{W}^{\eta\zeta}{}_{\kappa\xi}(T_\eta{}^\xi{}_{\zeta}-T_\zeta{}^\xi{}_{\eta})
+\frac12 g_{\mu\lambda} Ric^\eta{}_\xi T_\eta{}^\xi{}_{\kappa}
\right)
\\
&+&d_7 \left(
Ric^\eta{}_\lambda\left(
T_{\eta\mu\kappa} - T_{\kappa\mu\eta}
\right)
+\frac12 g_{\kappa\mu} \mathcal{W}^{\eta\zeta}{}_{\lambda\xi}(T_\eta{}^\xi{}_{\zeta}-T_\zeta{}^\xi{}_{\eta})
+\frac12 g_{\kappa\mu}Ric^\eta{}_\xi T_\eta{}^\xi{}_{\lambda}
)
\right)
\\
&+&b_{10} \left(T_\eta{}^\xi{}_{\zeta}-T_\zeta{}^\xi{}_{\eta})(g_{\mu\kappa}\mathcal{W}^{\eta\zeta}{}_{\lambda\xi}-g_{\mu\lambda}\mathcal{W}^{\eta\zeta}{}_{\kappa\xi}\right)
\\
&+&2b_{10}\left(
\mathcal{W}^\eta{}_{\mu\kappa\xi}(T_\eta{}^\xi{}_{\lambda}-T_\lambda{}^\xi{}_{\eta} )
+ \mathcal{W}^\eta{}_{\mu\lambda\xi}(T_\kappa{}^\xi{}_{\eta}-T_\eta{}^\xi{}_{\kappa}) - \mathcal{W}^{\xi\eta}{}_{\kappa\lambda} T_{\eta\mu\xi}
\right)=0.
\end{eqnarray*}
This change in no way affects the proof of the main theorem, as this form of the field equation is even simpler and does not contain any additional terms, but we believed it was important to point out for purposes of future work. This mistake was noticed in the process of generalising the explicit form of
the field equations (\ref{eulerlagrangemetric}), (\ref{eulerlagrangeconnection}), i.e. the equations equation (26) and (27) from \cite{pasicvassiliev}, see \cite{skopje}.

The two papers of Singh \cite{singh1,singh2} are of particular interest to us, where the author constructs solutions for the Yang--Mills case (\ref{YMq}) with purely axial and purely trace torsion respectively, see (\ref{torsion-pieces1}), and unlike the solution of \cite{griffiths3},
$\{Ric\}$ is not assumed to be zero. It is obvious that these solutions differ
from the ones presented in our work, as the torsion of generalised pp-waves is assumed to be purely tensor.
It would however be of interest to us to see whether this construction of Singh's can be expanded
to our most general $\mathrm{O}(1,3)$-invariant quadratic form $q$ with 16 coupling
constants.

We also hope to see whether it is possible to produce torsion waves which are purely axial or trace and combine them with the pp-wave metric, in a similar fashion as was done with purely tensor torsion waves, in order to produce new solutions of quadratic metric--affine gravity and also give \emph{their} physical interpretation in the near future.


\begin{thebibliography}{74}

\bibitem{Adamowitz}
Adamowicz, W.: Plane waves in gauge theories of gravitation. Gen. Rel. Grav. {\bf 12}, 677–691
(1980)

\bibitem{Alekseevsky}
Alekseevsky, D. V.:
Holonomy groups and recurrent tensor fields in Lorentzian spaces.
 Problems of the Theory of Gravitation and Elementary Particles
issue 5 edited by Stanjukovich, K. P.
(Moscow: Atomizdat) 5--17. In Russian (1974)

\bibitem{audretsch2}
Audretsch, J.:
Asymptotic behaviour of neutrino fields in curved space-time.
 Commun. math. Phys.  {\bf 21}, 303--313 (1971)

\bibitem{audretsch3}
Audretsch, J.,  Graf, W.:
Neutrino radiation in gravitational fields.
Commun. math. Phys.  {\bf 19}, 315--326 (1970)

\bibitem{Babourova}
Babourova, O. V., Frolov, B. N., Klimova, E. A.: Plane torsion waves in quadratic gravitational theories in Riemann–Cartan space.
Class. Quantum Grav. {\bf 16}, 1149–1162 (1999) [arXiv:gr-qc/9805005]

\bibitem{LL4}
Berestetskii, V. B., Lifshitz, E. M.,  Pitaevskii, L. P.:  Quantum
Electrodynamics (Course of Theoretical Physics vol 4) 2nd edn,
Pergamon Press, Oxford, (1982)

\bibitem{blagojevic}
Blagojevic, M.:  Gravitation and Gauge Symmetries.
Institute of Physics Publishing, Bristol (2002)

\bibitem{blagojevichehl}
Blagojevi\' c, M., Hehl F.W.: Gauge Theories of
Gravitation. A Reader with Commentaries. Imperial College Press, London (2013)

\bibitem{wheeler}
Brill, D. R.,  Wheeler, J. A.: Interaction of Neutrinos and
Gravitational Fields.  Rev. Mod. Phys. {\bf 29}, 465--479 (1957)

\bibitem{brinkmann}
Brinkmann, M. W.: On Riemann spaces conformal to Euclidean space.
Proceedings of the National Academy of Sciences of USA {\bf 9} 1–-3 (1923)

\bibitem{Bryant}
Bryant, R. L.: Pseudo--Riemannian metrics with parallel spinor
fields and vanishing Ricci tensor.  Global Analysis and
Harmonic Analysis (Marseille--Luminy, 1999) S\'emin. Congr.
\textbf{4} (Paris: Soc. Math. France), 53--94 (2000) [arXiv:math/0004073]

\bibitem{buchbinderandkuzenko}
Buchbinder, I. L., Kuzenko, S. M.:  Ideas and Methods of
Supersymmetry and Supergravity. Institute of Physics Publishing,
Bristol (1998)

\bibitem{buchdahl}
Buchdahl, H.A.: Mathematical Reviews {\bf 20}, 1238 (1959)

\bibitem{collinsonmorris}
Collinson, C. D.,  Morris, P. B.:
Space-time admitting neutrino fields with zero energy and momentum.
 J. Phys. {\bf A6}, 915--916 (1972)

\bibitem{cotton}
Cotton, \' E.:  Sur les variétés a trois dimensions. Annales de la Faculté des Sciences de Toulouse. II {\bf 1 4} 385–-438 (1899)

\bibitem{davisray}
Davis, T. M., Ray, J. R.:
Ghost neutrinos in general relativity.
 Phys. Rev. D {\bf 9}, 331--333 (1974)

\bibitem{eddington}
Eddington, A. S.: The Mathematical Theory of Relativity. 2nd ed. Cambridge : The University Press (1952)

\bibitem{Esser}
Esser, W.: Exact Solutions of the Metric--Affine Gauge Theory
of Gravity. (University of Cologne: Diploma Thesis) (1996)

\bibitem{fairchild1976}
Fairchild, E. E. Jr.: Gauge theory of gravitation. Phys. Rev. D {\bf 14}, 384--391 (1976)

\bibitem{fairchild1976erratum}
Fairchild, E. E. Jr.:  Erratum: Gauge theory of gravitation. Phys. Rev. D {\bf 14}, 2833 (1976)

\bibitem{garcia}
Garc\' ia, A., Mac\' ias, A., Puetzfeld, D., Socorro, J.: Plane fronted waves in metric
affine gravity. Phys. Rev. D {\bf 62}, 044021 (2000) [arXiv:gr-qc/0005038]

\bibitem{garciaCotton}
Garc\' ia, A., Hehl, F.W., Heinicke, C., Mac\' ias, A.:  The Cotton tensor in Riemannian spacetimes. Class. Quantum Grav. {\bf 21}, 1099–1118 (2004) [arXiv:gr-qc/0309008]


\bibitem{griffiths book}
Griffiths, J. B.: Colliding Plane Waves in General Relativity. Oxford University Press (1991)

%\bibitem{griffiths neutrino}
%Griffiths, J. B.: Neutrino fields in Einstein--Cartan theory.
%Gen. Rel. Grav. {\bf 13}, 227--237 (1981)

\bibitem{griffiths1970_1}
Griffiths, J. B.,  Newing, R. A.: The two-component neutrino field in general relativity.
 J. Phys. A {\bf 3}, 136--149 (1970)

\bibitem{griffiths1970_2}
Griffiths, J. B.,  Newing, R. A.: Tetrad equations for the two-component neutrino field in general relativity.
J. Phys. A {\bf 3}, 269--273 (1970)

\bibitem{griffiths1972_1}
Griffiths, J. B.: Some physical properties of neutrino-gravitational fields.
 Int. J. Theoret. Phys. {\bf 5}, 141--150 (1972)

\bibitem{griffiths1972_2}
Griffiths, J. B.: Gravitational radiation and neutrinos.
 Commun. Math. Phys. {\bf 28}, 295--299 (1972)

\bibitem{griffiths1980}
Griffiths, J. B.:Ghost neutrinos in Einstein--Cartan theory.
 Phys. Lett. A {\bf 75}, 441--442 (1980)

\bibitem{griffiths1981}
Griffiths, J. B.:  Neutrino fields in Einstein--Cartan theory.
Gen. Rel. Grav. {\bf 13}, 227--237 (1981)

\bibitem{hehl habilitation}
Hehl, F. W.: Spin und Torsion in der Allgemeinen
Relativit\"atstheorie oder die Riemann--Cartansche Geometrie der
Welt. (Technischen Universit\"at Clausthal: Habilitationsschrift) (1970)

\bibitem{hehl neutrino 1}
Hehl, F. W.: Spin and torsion in general relativity I: Foundations.
 Gen. Rel. Grav. {\bf 4}, 333--349 (1973)

\bibitem{hehl neutrino 2}
Hehl, F. W.: Spin and torsion in general relativity II: Geometry and field equations.
 Gen. Rel. Grav. {\bf 5}, 491--516 (1974)

\bibitem{hehlreview}
Hehl, F. W., McCrea, J. D., Mielke, E. W.,  Ne'eman, Y.: Metric--affine
gauge theory of gravity: field equations, Noether identities, world
spinors, and breaking of dilation invariance.  Phys. Rep. {\bf
258}, 1--171 (1995) [arXiv:gr-qc/9402012]

\bibitem{hehlandmaciasexactsolutions2}
Hehl, F. W., Mac{\'\i}as, A.:  Metric--affine gauge theory of
gravity II. Exact solutions.  Int. J. Mod. Phys. {\bf D8}, 399--416 (1999) [arXiv:gr-qc/9902076]

\bibitem{hehl neutrino 3}
Hehl, F. W., von der Heyde, P., Kerlick, G. D.,  Nester, J. M.: General
relativity with spin and torsion: Foundations and prospects.
Rev. Mod. Phys. {\bf 48}, 393--416 (1976)

\bibitem{higgs}
Higgs, P. W.: Quadratic Lagrangians and general relativity. Nuovo Cimento {\bf 11}, 816--820 (1959)

\bibitem{King and Vassiliev}
King, A. D.,  Vassiliev, D.: Torsion waves in metric--affine field
theory.  Class. Quantum Grav. {\bf 18}, 2317--2329 (2001) [arXiv:gr-qc/0012046]

\bibitem{exact solutions Kramer}
Kramer, D., Stephani, H., Herlt, E.,  MacCallum, M.:  Exact
Solutions of Einstein's Field Equations. Cambridge University
Press, Cambridge (1980)

\bibitem{KronerLattices}
Kr\" oner, E.: Continuum theory of defects.
Physics of Defects, Les Houches, Session XXXV, 1980, R.Balian ct al.,
eds., North--Holland, Amsterdam (1980)

\bibitem{KronerKristal}
Kr\" oner, E.:  The continuized crystal -- a bridge between micro-- and macromechanics.
Gesellschaft angewandte Mathematik und Mechanik Jahrestagung Goettingen West Germany Zeitschrift Flugwissenschaften, Vol. {\bf 66} (1986)

\bibitem{kuchowitz}
Kuchowicz, C., \. Zebrowski, J.:
The presence of torsion enables a metric to allow a gravitational field.
Phys. Lett. A {\bf 67}, 16--18 (1978)

\bibitem{lanczos1}
Lanczos, C.:  A remarkable property of the Riemann--Christoffel tensor in four dimensions.
Ann. Math. {\bf 39}, 842--850 (1938)

\bibitem{lanczos2}
Lanczos, C.: Lagrangian multiplier and Riemannian spaces. Rev. Mod. Phys. {\bf 21}, 497--502 (1949)

\bibitem{lanczos3}
Lanczos, C.: Electricity and general relativity. Rev. Mod. Phys. {\bf 29}, 337--350 (1957)

\bibitem{LL2}
Landau, L. D.,  Lifshitz, E. M.:  The Classical Theory of Fields
(Course of Theoretical Physics vol 2) 4nd edn. Pergamon Press,
Oxford (1975)

\bibitem{mielkepseudoparticle}
Mielke, E. W.: On pseudoparticle solutions in Yang's theory of
gravity. Gen. Rel. Grav. {\bf 13}, 175--187 (1981)

\bibitem{Nakahara}
Nakahara, M.:  Geometry, Topology and Physics. Institute of Physics Publishing, Bristol (1998)

\bibitem{obukhov1}
Obukhov, Y. N.: Generalized plane fronted gravitational waves in any dimension. Phys. Rev. D {\bf 69}, 024013 (2004) [gr-qc/0310121]

\bibitem{obukhov2}
Obukhov, Y. N.: Plane waves in metric-affine gravity. Phys. Rev. D {\bf 73}, 024025 (2006) [gr-qc/0601074]

\bibitem{olesen}
Olesen, P.:  A relation between the Einstein and the Yang--Mills
field equations. Phys. Lett. B {\bf 71}, 189--190 (1977)

\bibitem{balkanica}
Pasic, V.: New Vacuum Solutions for Quadratic Metric-Affine Gravity - a Metric Affine Model for the Massless Neutrino?
 Mathematica Balkanica New Series {\bf 24}, Fasc 3-4 329 (2010)

\bibitem{skopje}
Pasic, V., Barakovic, E., Okicic, N.: A new representation of the field equations of quadratic metric-affine gravity.
Adv. Math., Sci. J. {\bf 3} 1, 33--46 (2014)

\bibitem{pasicvassiliev}
Pasic, V.,  Vassiliev, D.:
PP--waves with torsion and metric--affine gravity.
 Class. Quantum Grav. {\bf 22}, 3961--3975 (2005) [arXiv:gr-qc/0505157]

\bibitem{pauli}
Pauli, W.: Zur Theorie der Gravitation und der Elektrizit\" at von
Hermann Weyl. Physik. Zaitschr. {\bf 20}, 457--467 (1919)

\bibitem{pavelleApr1975}
Pavelle, R.: Unphysical solutions of Yang's gravitational--field
equations. Phys. Rev. Lett. {\bf 34}, 1114 (1975)

\bibitem{penroserindler}
Penrose, O.,  Rindler, W.: Propagating modes in gauge field theories
of gravity, 2 volumes.
Cambridge University Press, Oxford (1984, 1986)

\bibitem{peres}
Peres, A.: Some gravitational waves.  Phys. Rev. Lett.
{\bf 3}, 571--572 (1959)

\bibitem{peresweb}
Peres, A.: PP -- WAVES {\it preprint} hep--th/0205040 (reprinting of
\cite{peres}) (2002)

\bibitem{pirani}
Pirani, F. A E.:  Introduction to Gravitational Radiation
Theory.  Lectures on General Realtivity. Prentice--Hall,
Inc. englewood Cliffs, New Jersey (1964)

\bibitem{singh1}
Singh, P.: On axial vector torsion in vacuum quadratic Poincar\' e gauge field theory.
 Phys. Lett. A {\bf 145}, 7--10 (1990)

\bibitem{singh2}
Singh, P.:  On null tratorial torsion in vacuum quadratic Poincar\' e gauge field theory.
Class. Quantum Grav. {\bf 7}, 2125--2130 (1990)

\bibitem{singhgriffiths}
Singh, P., Griffiths, J. B.: On neutrino fields in Einstein--Cartan theory.
 Phys. Lett. A {\bf 132}, 88--90 (1988)

\bibitem{griffiths3}
Singh, P.,  Griffiths, J. B.: A new class of exact solutions of the
vacuum quadratic Poincar\'e gauge field theory.  Gen. Rel. Grav.
{\bf 22}, 947--956 (1990)

\bibitem{stephenson}
Stephenson, G.: Quadratic Lagrangians and general relativity. Nuovo Cimento {\bf 9}, 263--269 (1958)

\bibitem{StreaterWightman}
Streater, R. F.,  Wightman, A. S.:  PCT, spin and statistics,
and all that. Princeton Landmarks in Physics (Princeton University
Press, Princeton, NJ), ISBN 0-691-07062-8, corrected third printing
of the 1978 edition (2000)

\bibitem{thompsonFeb1975}
Thompson, A. H.:  Yang's gravitational field equations. Phys.
Rev. Lett. {\bf 34}, 507--508 (1975)

\bibitem{thompsonAug1975}
Thompson, A. H.: Geometrically degenerate solutions of the
Kilmister--Yang equations. Phys. Rev. Lett. {\bf 35}, 320--322 (1975)

\bibitem{pseudo}
Vassiliev, D.:  Pseudoinstantons in metric--affine field theory.
Gen. Rel. Grav. {\bf 34}, 1239--1265 (2002) [arXiv:gr-qc/0108028]

\bibitem{garda}
Vassiliev, D.: Pseudoinstantons in metric--affine field theory.
 Quark Confinement and the Hadron Spectrum V, edited by
Brambilla, N. and Prosperi, G. M. (Singapore: World Scientific) 273--275 (2003)

\bibitem{poland}
Vassiliev, D.: Quadratic non--Riemannian gravity.  Journal of
Nonlinear Mathematical Physics {\bf 11}, Supplement, 204--216 (2004)

\bibitem{annalen}
Vassiliev, D.: Quadratic metric--affine gravity.  Ann. Phys.
(Lpz.) {\bf 14}, 231--252 (2005) [arXiv:gr-qc/0304028]

\bibitem{weylquadraticaction}
Weyl, H.:  Eine neue Erweiterung der Relativit\" atstheorie. Ann. Phys. (Lpz.) \textbf{59}, 101--133 (1919)

\bibitem{wilczek}
Wilczek, F.: Geometry and Interaction of Instantons. In: Quark Confinement and Field theory, eds.
D.~R.~Stump and D.~H.~Weingarten, Wiley-Interscience, New York, 211--219 (1977)

\bibitem{yang}
Yang, C.N.: Integral Formalism for Gauge Fields. Phys. Rev.
Lett. {\bf 33}, 445--447 (1974)

\end{thebibliography}
\end{document}